\documentclass[12pt]{article}

\usepackage{amsmath, amsthm, amssymb, latexsym}
\usepackage{mathtools}
\usepackage{algorithm}
\usepackage{algpseudocode}
\usepackage{multirow,enumitem,url,booktabs,xcolor,bbm,soul,natbib,caption,subcaption}
\usepackage{float} 
\usepackage{bm}
\usepackage{graphicx} 
\usepackage{appendix}
\RequirePackage[colorlinks,citecolor=blue,urlcolor=blue]{hyperref}
\usepackage{bookmark}
\algnewcommand\Input{\item[\textbf{Input:}]} 
\floatstyle{ruled}
\newfloat{procedure}{tbp}{lop}
\floatname{procedure}{Procedure}

\newtheorem{theorem}{Theorem}

\newtheorem{definition}{Definition}

\newtheorem{proposition}{Proposition}
\allowdisplaybreaks
\def\Bernoulli {\text{Bernoulli}}
\def\E {\text{E}}
\def\I {\text{I}}
\def\vech {\text{vech}}

\def\vec{\text{\rm vec}}

\def\Cov{\text{\rm Cov}}
\def\Cor{\text{\rm Cor}}

\def\KL{\text{\rm KL}}

\def\E{\text{\mathbb E}}

\def\diag{\text{\rm diag}}
\def\blockdiag{\text{\rm blockdiag}}

\def\vech{\text{\rm vech}}

\newcommand\Poisson{\text{Poisson}}
\def\Q{{\mathcal{Q}}}
\def\L{{\mathcal{L}}}
\def\G{{\mathcal{G}}}
\def\T{{\mathcal{T}}}
\def\H{{\mathcal{H}}}

\def\N{{\text{N}}}
\def\IG{{\text{IG}}}
\def\E{{\text{E}}}
\def\SN{{\text{SN}}}
\def\IG{{\text{IG}}}
\def\NG{{\text{NG}}}

\def\Exp{{\text{Exp}}}

\def\bfone{{\bf 1}}

\def\P{\text{P}}

\graphicspath{{img/}}
\newcommand{\anon}{1}

\usepackage{xcite, xr}

\begin{document}

\def\spacingset#1{\renewcommand{\baselinestretch}%
{#1}\small\normalsize} \spacingset{1}


\if1\anon
{
  \title{\bf Generalized reparametrized variational Bayes \\
  with skew-symmetric normalization}
  \author{Aoxiang Chen (aoxiang.chen@u.nus.edu)\hspace{.2cm}\\
    Department of Statistics and Data Science, National University of Singapore\\
    and \\
    Linda S. L. Tan\thanks{Linda Tan's research is supported by the Ministry of Education, Singapore, under its Academic Research Fund Tier 2 (Award MOE-T2EP20222-0002).} (statsll@nus.edu.sg)\\
    Department of Statistics and Data Science, National University of Singapore}
  \maketitle
} \fi

\if0\anon
{
  \bigskip
  \bigskip
  \bigskip
  \begin{center}
    {\LARGE\bf Generalized reparametrized variational Bayes with skew-symmetric normalization}
  \end{center}
  \medskip
} \fi

\bigskip
\begin{abstract}
Bayesian hierarchical models with high-dimensional latent structure require scalable posterior approximations that preserve key dependencies while remaining computationally tractable. Mean-field variational inference (MFVI) is efficient, but can be unreliable when local variables are strongly correlated or tightly coupled to global variables. We propose KNorm-RVB, a generalized reparametrized variational Bayes framework for latent Gaussian and latent non-Gaussian models with sparse local precision matrices. KNorm-RVB maps the conditional posterior of local variables toward a standard Gaussian via normalization followed by skewness reduction, enabled by a novel $K$-component skew-symmetric density representation. This reparametrization centers the transformed conditional local posterior at an optimized reflection point and decorrelates local and global variables, making MFVI much more effective. Under symmetry conditions, we show that MFVI recovers the local posterior mean and correlation matrix exactly, motivating KNorm-RVB’s normalization and symmetrization of the conditional local posterior before applying MFVI. We combine a Gaussian variational family for reparametrized local variables with a flexible closed skew normal family for the remaining variables. Across generalized linear mixed models, mixed multinomial logit models, spatial autoregressive models, and stochastic volatility models, KNorm-RVB improves posterior approximation accuracy over existing methods.
\end{abstract}

\noindent%
\vfill

\newpage
\spacingset{1.8} 

\section{Introduction}
Bayesian hierarchical models are powerful tools for data analysis, combining global variables that capture population-level structure with local latent variables representing subject, area or time specific variation \citep{Gelman2006}. They are widely used in political science \citep{Park2004}, epidemiology \citep{Blangiardo2013} and ecology \citep{Royle2008}, but posterior inference becomes challenging when latent variables are high-dimensional and strongly correlated. Markov chain Monte Carlo (MCMC) methods, including Hamiltonian Monte Carlo \citep[HMC,][]{Hoffman2014}, provide asymptotically exact inference under standard regularity conditions \citep{Robert1999}, yet can mix poorly in such settings. This motivates scalable posterior approximations, especially when models must be fit repeatedly for model selection.

Integrated nested Laplace approximation \citep[INLA,][]{Rue2009, Rue2017} is widely used for latent Gaussian models. By exploiting Gaussian Markov random field sparsity and combining Laplace approximations with numerical integration over global variables, INLA can efficiently approximate posterior marginals, with recent work further improving its accuracy \citep{Chiuchiolo2022, Van2024, Dutta2026}. Variational inference \citep[VI,][]{Blei2017}, by contrast, approximates the joint posterior with a tractable family and can scale to models with many global and local variables, though accuracy depends on how well the variational family matches the true posterior. Full-rank Gaussian VI is flexible but can be unstable or expensive in high dimensions \citep{Ko2024}, whereas mean-field VI (MFVI) scales well but imposes posterior independence among variables and often underestimates uncertainty \citep{Neville2014}. Structured variants, including sparse precision Gaussian VI \citep{Tan2018}, conditionally structured VI \citep{Tan2020B}, amortized VI \citep{Agrawal2021}, hybrid VI \citep{Loaiza2022} and partially factorized VI \citep{Goplerud2025}, aim to better balance accuracy and scalability.

In this work, we improve the accuracy and efficiency of MFVI for hierarchical models through {\em reparametrizations}, which are known to accelerate MCMC by reshaping posterior geometry \citep{Gelfand1995, Gelfand1996}. Hierarchical models are often expressed in a centered parameterization, where local variables depend directly on global variables, or noncentered parameterization, where local variables are a priori independent of the global variables \citep{Betancourt2015}. The better choice is data dependent and hard to predict in advance. Partial noncentering interpolates between these extremes and can outperform both \citep{Papaspiliopoulos2007, Bass2019}. Related approaches include the ancillarity-sufficiency interweaving strategy \citep{Yu2011}, and transport maps that decouple and Gaussianize the target prior to HMC sampling  \citep{Osmundsen2021}.

Reparametrization can similarly improve VI. \citet{Tan2013} develop partial noncentering for generalized linear mixed models (GLMMs), improving both convergence and approximation accuracy. \citet{Tan2021} propose {\em reparametrized variational Bayes} (RVB), a generalization of partial noncentering that applies MFVI after an affine transformation of the local variables to approximately normalize them (zero mean, identity covariance), and weaken posterior dependence on global variables. Marginally augmented variational Bayes \citep{Goplerud2022} goes beyond MFVI via Bayesian post-processing, inducing dependence by sampling expansion parameters introduced through random effects transformations. \cite{Chen2025rotate} perform MFVI in a rotated coordinate system, obtained by principal component analysis, to reduce posterior dependencies among variables.

Here we generalize RVB by addressing two key limitations. First, the existing RVB is tailored to latent Gaussian models with conditionally independent local variables given the global variables, and it does not accommodate conditionally correlated local structure, or latent non-Gaussian models that can capture jumps or spikes. Second, transformed local variables may remain skewed after normalization, making the Gaussian approximation in RVB inadequate. We address both issues by extending RVB to a broader class of hierarchical models, augmenting normalization with skewness reduction, and using non-Gaussian variational families for the global variables to better represent asymmetric posteriors.

Our contributions are fourfold. First, we develop a generalized RVB framework for latent Gaussian models with sparse precision matrices that allow conditional dependence among local variables, and extend it to latent non-Gaussian models by replacing Gaussian driving noise with a heavy-tail alternative \citep{Cabral2024}. Second, we propose KNorm-RVB, which reduces residual skewness after curvature-based normalization of the local variables, using a $K$-component skew-symmetric representation that generalizes the two-component formulation of \cite{Wang2004}. We also derive a novel stochastic representation that allows simulation from the symmetrized density, enabling analytic evaluation of the reparametrized joint density. Third, we use a richer variational family, the closed skew normal (CSN) subclass \citep{Tan2025}, for global and mixing variables, while retaining Gaussian approximation for transformed local variables. Fourth, we establish theory, showing that KNorm-RVB centers the transformed conditional local posterior at an optimized reflection point and decorrelates local and global variables, thereby improving the effectiveness of MFVI. Under symmetry conditions, MFVI also recovers the local posterior mean and correlation matrix exactly, motivating KNorm-RVB’s normalization and symmetrization of the conditional local posterior before applying MFVI, and extending recent work on symmetry guarantees in VI \citep{Margossian2025b, Margossian2025a, Marks2026}.

Section \ref{sec: RVB review} reviews the RVB approach, and Section \ref{sec:g_RVB} extends it to latent Gaussian and latent non-Gaussian models. Section \ref{sec:skewness_correction} generalizes the skew-symmetric density to the $K$-component setting, and shows how this representation improves the Gaussianization of normalized local variables by reducing residual skewness. Section \ref{sec:CSN} describes the variational families and optimization algorithm. Section \ref{sec:theory} studies the finite reflection symmetry induced by KNorm-RVB and its implications for mean recovery and posterior dependence structure. Section \ref{sec:app} reports experimental results and Section \ref{sec:conclusion} concludes with a discussion.

\section{Review of reparametrized variational Bayes}\label{sec: RVB review}
Consider a hierarchical model for observed data $y = (y_1^\top, \dots, y_n^\top)^{\top}$, with global variables $\theta_g$ and local (latent) variables $b = (b_1^\top, \dots, b_n^\top)^{\top} \in \mathbbm{R}^N$, where $b_i$ is associated with $y_i$ for $i=1, \dots, n$. Let $\theta = (\theta_g^\top, b^\top)^\top$ and suppose the joint density factorizes as
\begin{equation} \label{hierarchical model}
p(y,\theta) = p(\theta_g) p(b\mid\theta_g) \prod_{i=1}^n p(y_i\mid b_i,\theta_g),
\end{equation}
where $p(\theta_g)$ is the prior on $\theta_g$, $ p(b\mid\theta_g)$ is the conditional density of $b$ given $\theta_g$ and $p(y_i\mid b_i,\theta_g)$ is the likelihood of observing $y_i$ given $b_i$ and $\theta_g$.

The posterior $p(\theta\mid y) =  p(y \mid \theta) p(\theta)/p(y)$ is often intractable because the marginal likelihood $p(y)=\int p(y \mid \theta) p(\theta) \,d\theta$ has no closed form. VI approximates $p(\theta\mid y)$ with a more tractable density $q_\lambda(\theta)$ from a variational family indexed by $\lambda$, chosen to minimize the Kullback-Leibler divergence (KLD), $\KL \{q_\lambda(\theta) \| p(\theta\mid y)\} = \int q_\lambda(\theta) \log \{q_\lambda(\theta) / p(\theta\mid y)\} d\theta$. Using the identity, $\log p(y) = \L(\lambda) + \KL\{q_\lambda(\theta) \| p(\theta\mid y)\}$, minimizing the KLD is equivalent to maximizing the evidence lower bound (ELBO),
\begin{equation*}
\L(\lambda) = \E_{q_\lambda(\theta)} \{ \log p(y,\theta) - \log q_\lambda(\theta) \},
\end{equation*} 
which avoids the intractable $\log p(y)$, and satisfies $\L(\lambda) \leq \log p(y)$ since KLD $\geq 0$.

For complex hierarchical models, $\L(\lambda)$ is rarely available in closed form, but it can be optimized by stochastic gradient ascent using updates, $\lambda \leftarrow \lambda + \widehat{\nabla}_\lambda \L(\lambda)$, where $\widehat{\nabla}_\lambda \L(\lambda)$ is an unbiased estimator of $\nabla_\lambda \L(\lambda)$. To reduce gradient variance, we use the reparametrization trick \citep{Kingma2015}. Rather than sampling $\theta \sim q_\lambda(\theta)$ directly, we set $\theta = t(\epsilon,\lambda)$ with $\epsilon \sim \pi(\epsilon)$, where $\pi(\epsilon)$ is independent of $\lambda$ and $t(\cdot, \cdot)$ is differentiable in $\lambda$. This yields $\L(\lambda) = \E_{\pi(\epsilon)} \{ \log p \left(y, t(\epsilon,\lambda) \right) - \log q_\lambda (t(\epsilon,\lambda)) \}$. Under standard regularity conditions, $\nabla_\lambda \L(\lambda) = \E_{\pi(\epsilon)} [ \nabla_\lambda \theta \, \{\nabla_\theta\log p(y, \theta) - \nabla_\theta \log q_\lambda (\theta) \} ]$, and an unbiased estimator follows by sampling $\epsilon \sim \pi(\epsilon)$ at each iteration and using automatic differentiation. We use the sticking-the-landing estimator \citep{Roeder2017}, which omits the score function term $\nabla_\lambda q_\lambda(\theta)$ with zero mean, and has desirable gradient variance properties \citep{Kim2024}.

\cite{Tan2021} further assumes that $p(b \mid \theta_g) = \prod_{i=1}^n p(b_i \mid \theta_g)$. Although $p(b \mid \theta_g, y) = \prod_{i=1}^n p(b_i \mid \theta_g, y_i)$, the mean-field variational approximation $q(\theta) = q(\theta_g)\prod_{i=1}^n q(b_i)$ can be overly restrictive as it cannot capture the typically strong posterior dependence between each $b_i$ and $\theta_g$. RVB mitigates this by applying an affine transformation on the local variables to reduce their posterior dependence on $\theta_g$. Concretely, if $p(b_i \mid \theta_g, y_i)$ is well approximated by $\N(\lambda_i, \Lambda_i)$, with $\lambda_i$ and $\Lambda_i$ depending on $\theta_g$, and $\Lambda_i = L_i L_i^\top$ is the Cholesky factorization, then the normalized variable $\tilde{b}_i = L_i^{-1}(b_i - \lambda_i) \approx \N(0, I)$. This suggests that $\tilde{b}_i$ is approximately independent of $\theta_g$ a posteriori. A mean-field variational approximation for the reparametrized model, $q(\tilde{\theta}) = q(\theta_g) \prod_{i=1}^n q(\tilde{b}_i)$, where $\tilde{\theta} = (\theta_g^\top, \tilde{b}^\top)^\top$ and $\tilde{b} = (\tilde{b}_1^\top, \dots, \tilde{b}_n^\top)^\top$, is therefore less restrictive and more accurate.

\cite{Tan2021} apply RVB to GLMMs by constructing a Gaussian approximation to $p(b_i \mid \theta_g, y_i)$, which determines how $\lambda_i$ and $L_i$ depend on $\theta_g$. They consider a second-order Taylor expansion of either (i) the likelihood $p(y_i \mid \eta_i)$ around an estimate $\hat{\eta}_i$ of the natural parameter $\eta_i$, or (ii) $\log p(b_i \mid \theta_g, y_i)$ about an estimate $\hat{b}_i$ of $b_i$. Empirically, the latter is more accurate when the data are weakly informative about the natural parameters. In their experiments, RVB converges rapidly and yields marginal posterior estimates comparable to INLA, while providing an explicit approximation to the full joint posterior. This performance is largely driven by the normalization step, which reparametrizes the local variables to have approximately mean zero with identity covariance a posteriori, simplifying initialization of local variational parameters and promoting fast, stable convergence.

\section{Generalized reparametrized variational Bayes}\label{sec:g_RVB}
Our model class is broader than that of \cite{Tan2021}, as we do not impose conditional independence of the local variables, $p(b \mid \theta_g) = \prod_{i=1}^n p(b_i \mid \theta_g)$. Thus, \eqref{hierarchical model} includes not only GLMMs but also models with structured dependence among the local variables, such as state space models (SSMs) and spatial autoregressive (SAR) models, which fall outside the scope of RVB. We first develop a generalized RVB framework for latent Gaussian models in Section \ref{sec:LGM}, and then extend it to latent non-Gaussian models in Section \ref{sec:lngm}.

\subsection{Latent Gaussian models}\label{sec:LGM}
In latent Gaussian models, we assume $b \sim \N(0, Q^{-1})$. Conditional independence among elements of $b$ are encoded by the precision matrix $Q$, with $Q_{ij} = 0$ implying that the $i$th and $j$th elements of $b$ are conditionally independent given the remaining elements. We assume $Q$ is a sparse matrix depending on $\theta_g$, as is typical in SSMs and SAR models.

As the local variables may be correlated a priori, $p(b \mid \theta_g,y)$ may not factorize as $\prod_{i=1}^n p(b_i \mid \theta_g,y_i)$. A mean-field variational approximation $q(\theta) = q(\theta_g)\prod_{i=1}^n q(b_i)$ is thus highly restrictive, as it ignores posterior dependence both between $\theta_g$ and $b$, and among the $\{b_i\}$ induced by $Q$. We address this limitation by reparametrizing the local variables to reduce these dependencies. In particular, 
\begin{equation}\label{eq:local_posterior}
\begin{aligned}
p(b\mid \theta_g, y) 
\propto \exp \left\{\sum_{i=1}^n \log p(y_i \mid b_i,\theta_g) - \frac{1}{2}b^\top Q b \right\},
\end{aligned}
\end{equation}
which is generally non-Gaussian due to the likelihood terms, despite the Gaussian prior.

To construct a transformation that normalizes $b$, we form a Gaussian approximation to $p(b\mid \theta_g, y)$ that preserves the sparsity of $Q$ and is accurate in regions of high posterior mass. A second-order Taylor expansion of $\log p(y_i \mid b_i, \theta_g)$ about an estimate $\hat{b}_i$ of $b_i$ gives
\begin{equation}\label{eq:second_taylor}
\log p(y_i \mid b_i,\theta_g) \approx \log p(y_i \mid \hat{b}_i,\theta_g) + (b_i-\hat b_i)^\top g_i(\hat b_i) - (b_i-\hat b_i)^\top H_i(\hat{b}_i)(b_i-\hat b_i)/2,
\end{equation}
where $g_i(b_i) = \nabla_{b_i} \log p(y_i  \mid  b_i, \theta_g)$ and $H_i(b_i) = -\nabla_{b_i}^2 \log p(y_i  \mid  b_i, \theta_g)$. We further define  
\begin{equation}\label{eq:g&H}
    g(b)=(g_1(b_1)^\top,\dots,g_n(b_n)^\top)^\top
    \quad \text{and} \quad
    H(b)=\blockdiag\{H_1(b_1),\dots,H_n(b_n)\}.
\end{equation}
Let $\hat b$ be the conditional posterior mode, satisfying $\nabla_b\log p(b \mid \theta_g,y) = g(b) - Q b = 0$, which we compute using Newton’s method. Substituting \eqref{eq:second_taylor} into \eqref{eq:local_posterior} yields 
\begin{equation*}
\begin{aligned}
p(b \mid \theta_g,y)&\propto \exp \left\{ b^\top g(\hat b) - b^\top H(\hat b)b/2 + b^\top H(\hat b)\hat b - b^\top Qb/2 \right\} \\
&=\exp\left[ b^\top \{Q+ H(\hat b)\}\hat b - b^\top \{Q+H(\hat b)\} b/2 \right],
\end{aligned}
\end{equation*}
where the last line uses $g(\hat b) = Q \hat b$. Hence 
\begin{equation*}
b \mid \theta_g,y \sim \N \left( \hat b, \{Q + H(\hat b)\}^{-1} \right) \quad \text{approximately}.
\end{equation*}
Importantly, the precision matrix $Q + H(\hat b)$ remains sparse because $H(\hat b)$ is block diagonal.

Let $Q + H(\hat b) = L L^\top$ be the Cholesky factorization, where $L$ is lower triangular and inherits the sparsity of $Q + H(\hat b)$. Define the affine transformation,
\begin{equation*}
\tilde b = L^{\top} (b-\hat b) \approx \N(0, I),
\end{equation*}
which suggests that $\tilde b$ is approximately independent of $\theta_g$ and uncorrelated across components a posteriori. For the reparametrized model with variables $\tilde\theta=(\theta_g^\top,\tilde b^\top)^\top$ where $\tilde b = (\tilde b_1, \dots, \tilde b_n)^\top$, we consider the factorized variational approximation,
\begin{equation*}
q_\lambda(\tilde{\theta})= q(\theta_g) \prod_{i} q(\tilde{b}_{i}).
\end{equation*}

The joint density of the reparametrized model follows from a standard change of variables. Conditional on $\theta_g$, the inverse map is $b = L^{-\top} \tilde{b} + \hat{b}$ with Jacobian $\partial b/\partial \tilde{b} = L^{-\top}$. It follows that $p(\tilde{b} \mid \theta_g)  = p(b \mid \theta_g) |\partial b/\partial \tilde{b}|= p(b \mid \theta_g) |L|^{-1}$. Therefore
\begin{align}\label{eq:RVB_logp_density}
\log p(y,\tilde\theta) = \log p(\theta_g) + \log p(b \mid \theta_g) + \sum_{i=1}^n \log p(y_i \mid b_i, \theta_g) - \log |L|,
\end{align}
where $b = L^{-\top} \tilde{b} + \hat{b}$ is substituted on the right-hand side.

We highlight two key differences from \cite{Tan2021}. First, since we do not assume $p(b \mid \theta_g) = \prod_{i=1}^n p(b_i \mid \theta_g)$, the latent field is normalized jointly in $b$ to preserve its prior dependence structure. Second, to exploit sparsity in $Q$, the normalization is based on Cholesky factorization of the precision $Q + H(\hat b)$, rather than per-subject covariance factorizations.

\subsection{Latent non-Gaussian models}\label{sec:lngm}
Latent Gaussian models can be inadequate for data with outliers or abrupt changes due to their light tails \citep{Walder2020, Cabral2023}. In such settings, Gaussian priors may oversmooth and inflate variance components, compromising uncertainty quantification, whereas heavy-tail priors provide robustness to extreme local deviations (e.g. sudden jumps or spikes). Following \citet{Cabral2024}, we extend latent Gaussian models by replacing the Gaussian driving noise in the latent field with a heavy-tail alternative.

Let $b^G \sim \N(0, Q^{-1})$ with $Q = D^\top D$ for a predetermined matrix $D$ that depends on $\theta_g$. Then $D b^G = Z$ for $Z \sim \N(0, I)$. We obtain a non-Gaussian prior for $b$ by replacing the Gaussian driving noise $Z$ by $\epsilon$ so that $D b = \epsilon$. We assume $\epsilon$ has independent components $\epsilon_{ij}$, each following a symmetric normal-inverse Gaussian (NIG) distribution with mean 0, variance 1 and excess kurtosis $3\eta$, where $\eta > 0$ controls tail heaviness: $\eta \rightarrow 0$ recovers Gaussian noise, and larger $\eta$ yields heavier tails. Using the variance-mean mixture representation of the NIG distribution, $\epsilon_{ij} \mid v_{ij} \sim \N(0, v_{ij})$ and $v_{ij} \mid \eta \sim \IG(1, \eta^{-1})$, where IG denotes the inverse Gaussian distribution, the induced prior for $b$ can be expressed as
\[
b \mid D, v \sim \N(0, D^{-1} \diag(v) D^{-\top}), \quad 
v_{ij} \mid \eta \stackrel{\text{iid}}{\sim} \IG(1, \eta^{-1}).
\]
The mixing variables are $v = (v_1^\top, \dots, v_n^\top)^\top$, where $v_{ij}$ is the $j$th element of $v_i$. We place an exponential prior on the tail parameter $\eta \sim \Exp(\alpha_\eta)$ and set $\alpha_\eta=1$ to mitigate overfitting \citep{Cabral2023}. This yields a robust latent non-Gaussian model that captures extreme local deviations, while retaining the efficiency of a sparse precision structure.

Let $\theta = (b^\top, \theta_g^\top, v^\top, \eta)^\top$ denote variables in the augmented latent non-Gaussian model, whose joint density factorizes as 
\begin{equation*}\label{eq:non_gaussian}
p(y,\theta) = p(\eta) p(v\mid\eta) p(\theta_g) p(b\mid\theta_g,v) 
\prod_{i=1}^n p(y_i\mid b_i,\theta_g).
\end{equation*}
\citet{Cabral2024} develop MFVI for this model in R-INLA, using $q(\theta) = q(b,\theta_g) q(v) q(\eta)$ and maximizing the ELBO via coordinate ascent. A key advantage is that updating $q(b,\theta_g)$ amounts to fitting a latent Gaussian model in INLA. They also propose a collapsed variant that integrates out $\eta$ via $p(v) = \int p(\eta)p(v \mid \eta) d\eta$, yielding the approximation $q(b, \theta_g)q(v)$.

The MFVI scheme above enforces posterior independence between $b$ and $v$, even though $b$ depends on $v$ through $p(b\mid \theta_g,v)$. To mitigate the mismatch, we adopt the RVB approach and reparametrize the local variables so that they are approximately independent of $(\theta_g, v)$ a posteriori. For the latent non-Gaussian model, 
\begin{equation*}\label{eq:local_posterior_NG}
\begin{aligned}
p(b\mid \theta_g, v, y) 
&\propto \exp\!\left\{\sum_{i=1}^n \log p(y_i \mid b_i,\theta_g) - \frac{1}{2}b^\top Q_\NG b \right\},
\end{aligned}
\end{equation*}
where $Q_\NG = D^\top \diag(v)^{-1} D$ inherits sparsity from $D$. Let $\hat b_{\NG}$ be the mode of $p(b\mid \theta_g, v, y)$ obtained by Newton’s method. A second-order Taylor expansion of $\log p(y_i \mid b_i, \theta_g)$ about $\hat b_{\NG}$ as in Section \ref{sec:LGM} yields $b\mid \theta_g, v, y \sim \N(\hat b_{\NG}, \{Q_\NG + H(\hat b_{\NG})\}^{-1})$ approximately, with $H(\cdot)$ defined in \eqref{eq:g&H}. Let $L_{\NG}$ be the Cholesky factor of the precision $Q_\NG + H(\hat b_{\NG})$. Then 
\[
\tilde b = L_{\NG}^{\top}\bigl(b-\hat b_{\NG}\bigr) \approx \N(0, I),
\]
so $\tilde b$ is a posteriori independent of $(\theta_g, v)$ and has independent components approximately.

For the reparametrized model with $\tilde\theta = (\tilde b^\top, \theta_g^\top, v^\top, \eta)^\top$, we consider MFVI with 
\begin{align*} 
q_\lambda(\tilde{\theta}) = q(\eta) q(\theta_g) \prod_{i} q(\tilde b_{i}) \prod_j q(v_{ij}) ,
\end{align*}
where $\tilde b = (\tilde b_1, \dots, \tilde b_n)^\top$. The log joint density of the reparametrized model follows from a standard change of variables, analogous to \eqref{eq:RVB_logp_density}.

In \cite{Cabral2024}, $q(\theta) = q(b,\theta_g) q(v) q(\eta)$ implies $q(v) = \prod_{i,j} q(v_{ij})$, and the optimal $q(\eta)$ and $q(v_{ij})$ lie in the generalized inverse Gaussian  family. Under our reparametrization, $v$ becomes non-separably coupled with other variables, so these results no longer apply. We therefore impose $q(v) = \prod_{i,j} q(v_{ij})$ explicitly and specify tractable variational families for $q(v_{ij})$ and $q(\eta)$ in Section \ref{sec:CSN}. Although $\eta$ can be marginalized as in \cite{Cabral2024}, we do not do so for simplicity and optimization stability. Marginalizing $\eta$ may further accelerate convergence and reduce underestimation of uncertainty.

\section{Improving Gaussianization via skewness reduction} \label{sec:skewness_correction}
In hierarchical models, local conditional posteriors may remain non-Gaussian even as $n$ grows, since each $b_i$ is informed mainly by $y_i$. Non-Gaussian likelihoods, near-boundary data, and weak within-unit signal can further leave $p(b\mid\theta_g,y)$ skewed or heavy-tailed. In Section \ref{sec:LGM}, an affine transformation maps $b$ to normalized variables $\tilde{b}$ that are closer to standard Gaussian, but residual skewness and tail mismatch may persist. Here, we further Gaussianize $\tilde{b}$ by reducing skewness, producing $b^*$. Specifically, we represent $p(\tilde{b} \mid \theta_g, y)$ as a $K$-component skew-symmetric density and derive a stochastic representation that enables simulation of $b^*$ from the symmetrized density. This yields a closed form reparametrized joint density, allowing ELBO optimization via the reparametrization trick. We refer to the approach as KNorm-RVB, as it augments RVB's normalization with skewness reduction via a $K$-component skew-symmetric representation. We first review skew-symmetric densities, introduce the $K$-component extensions and then derive the reparametrized log joint density analytically. We present the method for latent Gaussian models, and the extension to latent non-Gaussian models follows by additionally conditioning on mixing variables $v$.

\subsection{Skew-symmetric representation and construction} \label{sec:skewness_approx of local cond post}
Let $\hat{\theta} \in \mathbb{R}^d$ be fixed. \cite{Wang2004} establish two complementary results on constructing and representing skew-symmetric densities. First, let $f_{\hat\theta}: \mathbb{R}^d \rightarrow \mathbb{R}_+$ be a probability density function (pdf) that is symmetric about $\hat\theta$, which means $f(\theta) = f(2\hat{\theta} - \theta)$, and $w_{\hat\theta}:\mathbb{R}^d\rightarrow[0,1]$ be a skewing function satisfying $w_{\hat\theta}(\theta)+w_{\hat\theta}(2\hat\theta-\theta)=1$. Then  
\begin{equation*} \label{eq: two-skew-symm}
g(\theta) = 2 f_{\hat\theta}(\theta) w_{\hat\theta}(\theta),
\end{equation*}
is a pdf obtained by a {\em skew-symmetric construction}, where skewness is introduced via $w_{\hat\theta}(\cdot)$ without additional parameters. Independent samples from $g$ can be obtained by a rejection-free sampler: draw $x \sim f_{\hat\theta}$ and $u \sim \text{Unif}[0,1]$ independently, and set $\theta = x$ if $u \leq w_{\hat\theta}(x)$, and $\theta = 2\hat{\theta} - x$ otherwise. Second, any pdf $g: \mathbb{R}^d \rightarrow \mathbb{R}_+$ admits a unique {\em skew-symmetric representation} of the form $g(\theta) = 2 f_{\hat\theta}(\theta) w_{\hat\theta}(\theta)$, where $f_{\hat\theta} (\theta) = \{g(\theta) + g(2\hat{\theta} - \theta) \}/2$ is a pdf symmetric about $\hat\theta$ and $w_{\hat\theta}(\theta) = g(\theta)/\{g(\theta) + g(2\hat{\theta} - \theta) \}$ is a skewing function.

\subsection{$K$-component extensions} \label{sec:K_skew_rep}
The skew-symmetric representation of \cite{Wang2004} symmetrizes $g$ by averaging two components, $g(\theta)$ and $g(2\hat\theta-\theta)$. We extend this idea to a $K$-component skew-symmetric representation by averaging over $2\leq K\leq 2^d$ components. This can be viewed as the finite group special case of the symmetry modulation framework \citep{Jupp2016}.

\begin{figure}[b!]
\centering
\begin{subfigure}[b]{0.4\textwidth}
    \centering
    \includegraphics[width=\textwidth]{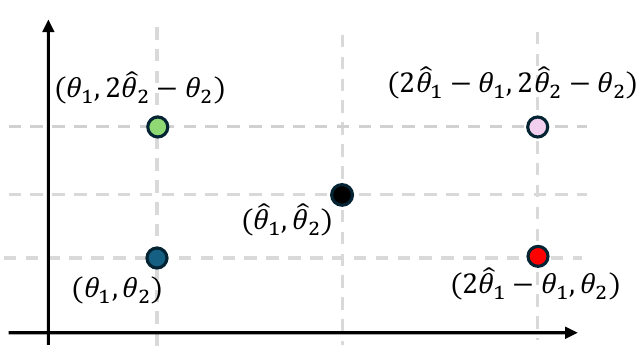}
\end{subfigure}
\hfill
\begin{subfigure}[b]{0.53\textwidth}
    \centering
    \includegraphics[width=\textwidth]{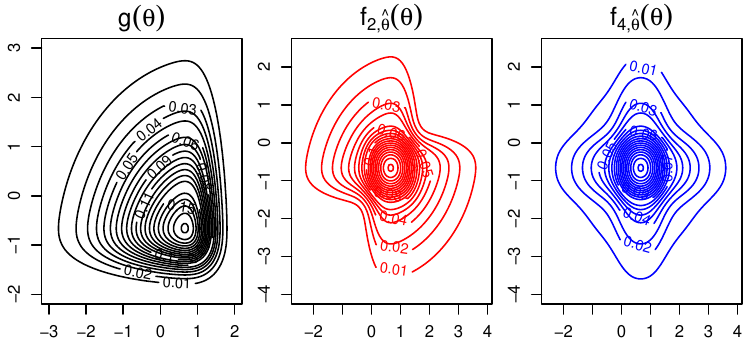}
\end{subfigure}
\caption{\small Left: Blockwise reflections about $\hat\theta$. Right: CSN density $g(\theta)$ with $\mu_\theta = 0$, $C_\theta = I_2$ and shape parameter $\gamma = (-5,5)^\top$, and symmetrized densities $f_{2,\hat\theta}(\theta)$ and $f_{4,\hat\theta}(\theta)$.}\label{Fig:skew2d}
\end{figure}

Figure \ref{Fig:skew2d} (left) depicts $\theta$ and its point-reflection $2\hat{\theta} - \theta$ for $d=2$, along with two points obtained by reflecting $\theta$ about $\hat{\theta}$ in one coordinate while holding the other fixed. Averaging over these four reflections gives the four-component symmetrized density,
\[
f_{4,\hat\theta} (\theta) = \frac{1}{4} \left\{ 
g \begin{pmatrix} \theta_1 \\ \theta_2 \end{pmatrix} + 
g \begin{pmatrix} 2\hat{\theta}_1 - \theta_1 \\ \theta_2 \end{pmatrix} + 
g \begin{pmatrix} \theta_1 \\ 2\hat{\theta}_2 - \theta_2 \end{pmatrix} + 
g \begin{pmatrix}  2\hat{\theta}_1 - \theta_1  \\ 2\hat{\theta}_2 - \theta_2 \end{pmatrix}
\right\}.
\]
To illustrate the effect of increasing $K$, let $g(\theta)$ be a highly asymmetric CSN density (details in Section \ref{sec:CSN}). Figure \ref{Fig:skew2d} (right) shows that $f_{2,\hat\theta}=\{g(\theta) + g(2\hat\theta-\theta)\}/2$ enforces only point symmetry about $\hat\theta$. In contrast, $f_{4,\hat\theta}$ also imposes coordinate-wise reflection symmetry about $\hat{\theta}$, producing a shape closer to an elliptical Gaussian-like form.

More generally, let $B_1,\dots,B_m$ be a partition of $\{1,\dots,d\}$ and $\varepsilon = (\varepsilon_1, \dots, \varepsilon_m)^\top$ with $\varepsilon_j\in\{-1,1\}$. Define $S(\varepsilon) = \blockdiag(\varepsilon_1 I_{|B_1|}, \dots, \varepsilon_m I_{|B_m|})$, where $|B_j|$ denotes the size of block $B_j$. For $m=1, \dots, d$, the induced {\em blockwise reflection group} is
\begin{align} \label{eq:refl_gp}
\G_m=\{S(\varepsilon):\varepsilon\in\{-1,1\}^m\}=\{S_1,\dots,S_K\},
\end{align}
a finite abelian group of orthogonal matrices under matrix multiplication, with identity $I$ and cardinality $K=|\G_m|=2^m$. We also define the trivial group $\G_0 = \{I_d\}$. The reflections induced by $\G_m$ about $\hat{\theta}$ are
\begin{align*}
T_{k, \hat{\theta}}(\theta) = \hat\theta+S_k(\theta-\hat\theta), \quad k=1,\dots,K.
\end{align*}
The set $\T_{\G_m,\hat\theta} = \{T_{1, \hat{\theta}},\dots,T_{K, \hat{\theta}}\}$ is a finite abelian group of bijections under composition, where $T_{j, \hat{\theta}}\circ T_{k, \hat{\theta}} (\theta) = \hat\theta + S_j S_k (\theta - \hat\theta)$. A pdf $f: \mathbb R^d \rightarrow \mathbb R_+$ is said to be $\T_{\G_m,\hat\theta}$-invariant if $f \{ T_{k, \hat{\theta}}(\theta) \} = f(\theta)$ $\forall \, \theta \in \mathbbm{R}^d$ and $k=1,\dots,K$. A function $w:\mathbb R^d\to[0,1]$ is called a $K$-component skewing function if $\sum_{k=1}^K w\{T_{k, \hat{\theta}}(\theta)\}=1$ $\forall \, \theta \in \mathbb R^d$.

\begin{proposition}[$K$-component skew-symmetric construction] \label{prop:K_skew_construction}
Let $f_{K,\hat\theta}: \mathbb R^d \rightarrow \mathbb R_+$ be a pdf that is $\T_{\G_m,\hat\theta}$-invariant for any $\hat\theta \in \mathbb R^d$ and $w_{K,\hat\theta}$ be a $K$-component skewing function. Then a function from $\mathbb R^d \rightarrow \mathbb R_+$ of the form 
$g(\theta) = K f_{K,\hat\theta}(\theta) w_{K,\hat\theta}(\theta)$
is a pdf.
\end{proposition}

\begin{proposition}[$K$-component skew-symmetric representation] \label{prop:K_skew_representation}
Let $\hat\theta$ be any point in $\mathbb R^d$, $\G_m$ be defined in \eqref{eq:refl_gp} and $g: \mathbb R^d \rightarrow \mathbb R_+$ be a pdf. Then $g$ can be expressed as 
\[
g(\theta) = K f_{K,\hat\theta}(\theta) w_{K,\hat\theta}(\theta),
\] 
where  $w_{K,\hat\theta}(\theta) = g(\theta) / \sum_{k=1}^K g\{T_{k, \hat{\theta}}(\theta)\}$ is a $K$-component skewing function, with $w_{K,\hat\theta}(\theta)=1/K$ if $\sum_{k=1}^K g\{T_{k, \hat{\theta}}(\theta)\}=0$ and $f_{K,\hat\theta}(\theta) = \frac{1}{K}\sum_{k=1}^K g \{T_{k, \hat{\theta}}(\theta)\}$ is a $\T_{\G_m,\hat\theta}$-invariant pdf. This representation is unique on the support of $g$.
\end{proposition}

Propositions \ref{prop:K_skew_construction} and \ref{prop:K_skew_representation} extend Propositions 1 and 3 of \citet{Wang2004} to $K$ components. Proposition \ref{prop:K_skew_construction} shows that a $\T_{\G_m,\hat\theta}$-invariant pdf multiplied by a $K$-component skewing function and rescaled by $K$ is a valid pdf. Proposition \ref{prop:K_skew_representation} establishes the converse. Given $\G_m$ and $\hat\theta$, any pdf admits a unique $K$-component skew-symmetric representation. The \citet{Wang2004} results are recovered when $m=1$ and $\G_1=\{I,-I\}$, while $m=d$ gives $K=2^d$, which is typically prohibitive for large $d$. Section \ref{sec: choose m} discusses practical choices of $m$.

Proposition \ref{prop: K_skew2symm} provides rejection-free samplers for (i) the $K$-component skew-symmetric density and (ii) its symmetrized $\T_{\G_m,\hat\theta}$-invariant density, generalizing as well as reversing the scheme of \cite{Wang2004}. To approximate a generic posterior $\pi(\theta) = p(\theta, y)/p(y)$ using a $K$-component skew-symmetric density, Theorem \ref{thm:K_optimal_skew_KL} characterizes the $K$-component skewing function that minimizes the KLD, extending results in \cite{Pozza2026} from the 2-component to the general $K$-component case.

\begin{proposition}[Stochastic representations]\label{prop: K_skew2symm}
Let $g(\theta) = K f_{K,\hat\theta}(\theta) w_{K,\hat\theta}(\theta)$ be the $K$-component skew-symmetric density constructed in Proposition \ref{prop:K_skew_construction}. 
\begin{enumerate}[label=(\roman*), nosep]
\item Sample $x \sim f_{K,\hat\theta}$. Conditional on $x$, draw $J \in \{1,\dots,K\}$ with $\Pr(J=k\mid x) = w_{K, \hat\theta} \{T_{k, \hat{\theta}}(x)\}$ for $k=1,\dots,K$. Then $\theta=T_{J, \hat{\theta}} (x)$ has pdf $g$.
\item Draw $\theta \sim g$ and $U \sim \text{Unif}\,\{1,\dots,K\}$ independently. Then $x=T_{U, \hat{\theta}} (\theta)$ has pdf $f_{K,\hat\theta}$.
\end{enumerate}
\end{proposition}

\begin{theorem}[Optimal skewing function] \label{thm:K_optimal_skew_KL}
Given reflection point $\hat{\theta}$, blockwise reflection group $\G_m$ and $\T_{\G_m,\hat\theta}$-invariant pdf $\bar q_{K,\hat\theta}(\theta)$, $\pi(\theta) = p(\theta, y)/p(y)$ is approximated by $q_{K,\hat\theta}(\theta) = K\bar q_{K,\hat\theta}(\theta) w_{K,\hat\theta}(\theta)$, where $w_{K,\hat\theta}(\theta)$ is a $K$-component skewing function. The $K$-component skewing function minimizing $\KL(\pi\| q_{K,\hat\theta})$ is 
\[
w^*_{K,\hat\theta}(\theta)= \frac{p(\theta,y)}{\sum_{k=1}^K p\{ T_{k, \hat{\theta}}(\theta), y\} }.
\]
\end{theorem}

\subsection{Reducing skewness in normalized local conditional posterior} \label{sec:K_skew_loc_cond_pos}
In Section \ref{sec:LGM}, we normalize the local variables via $\tilde b=L^\top(b-\hat b)$, where $\hat b$ is the mean of a Gaussian approximation to $p(b\mid\theta_g,y)$ with precision $Q+H(\hat b) = LL^\top$. Here, we seek to reduce residual skewness in $\tilde{b}$ by leveraging the $K$-component skew-symmetric representation of $p(\tilde b\mid\theta_g,y)$ ensured by Proposition \ref{prop:K_skew_representation}, for some reflection point $c$ and blockwise reflection group $\G_m$. Proposition \ref{prop: K_skew2symm} (ii) then yields a draw $b^*$ from the corresponding $\T_{\G_m, c}$-invariant density, and we later show that $b^*$ is closer to $\N(c, I)$ than $\tilde{b}$.

Adapting Section \ref{sec:K_skew_rep} with $\theta$ replaced by $\tilde b$, let $B_1,\dots,B_m$ be a partition of $\{1,\dots,N\}$ and $c\in\mathbb R^N$ denote the reflection point, whose choice is discussed in Section \ref{sec:CSN}. Procedure \ref{Algs:transformation_K} summarizes the construction of $b^*$ from $b$. Step 1 is deterministic and maps $b$ to $\tilde{b}$ via the affine transformation of Section \ref{sec:LGM}. Step 2 is stochastic and targets residual skewness. It treats $\tilde{b}$ as a draw from a $K$-component skew symmetric density as in Proposition \ref{prop:K_skew_representation}, and applies Proposition \ref{prop: K_skew2symm} to sample from the $\T_{\G_m, c}$-invariant pdf, $\frac{1}{K}\sum_{k=1}^K p_{\tilde{b} \mid \theta_g,y}\{T_{k,c}(\tilde{b})\}$, where $p_X(\cdot)$ denotes the pdf of $X$. Concretely, step 2 samples $S_*$ uniformly from $\G_m$ and applies the corresponding reflection to $\tilde{b}$. We refer to Procedure \ref{Algs:transformation_K} as KNorm-RVB, since it augments normalization with skewness reduction induced by a $K$-component skew-symmetric representation. The generalized RVB method of Section \ref{sec:g_RVB} is recovered when $\G_m = \G_0 =\{I_N\}$.

\begin{procedure}[tb!]
\caption{KNorm-RVB (normalization and $K$-component skewness reduction)}
\label{Algs:transformation_K}
\begin{algorithmic}[1]
\State Normalization: $\tilde b=L^\top(b-\hat b)$, where $\hat b$ is mode of $p(b\mid\theta_g,y)$ and $LL^\top = Q+H(\hat b)$.
\State Skewness reduction: Sample $S_* \sim \text{Unif}(\G_m)$ and set $b^* = c+S_*(\tilde b-c)$, for reflection point $c$ and blockwise reflection group $\G_m$.
\end{algorithmic}
\end{procedure}

Theorem \ref{thm: log_joint_density_skew_K} derives (i) the conditional density of $b^*$ and (ii) the reparametrized log joint density, by integrating out the uniformly distributed reflection index in step 2 of Procedure \ref{Algs:transformation_K}. These analytic expressions enable gradient computation by automatic differentiation and ELBO optimization via the reparametrization trick. Part (iii) shows that larger blockwise reflection groups yield stronger skewness reduction, as $b^*$ becomes closer to $\N(c,I)$ in KLD, where $\phi(\cdot\mid \mu, \Sigma)$ denotes the pdf of $\N(\mu, \Sigma)$. This motivates using a larger $K$ and aligns with Figure \ref{Fig:skew2d}, where averaging over four rather than two components produces a symmetrized density that is more nearly Gaussian. Part (iv) is a special case of (iii) and guarantees improvement from skewness reduction, showing that $b^*$ is always closer to $\N(c,I)$ than $\tilde b$.

\begin{theorem}[Effects of skewness reduction]\label{thm: log_joint_density_skew_K}
Let $b^*$ be obtained from $b$ via Procedure \ref{Algs:transformation_K}.
\begin{enumerate}[label=(\roman*), nosep]
\item $p(b^* \mid \theta_g, y) = \frac{1}{K|L|} \sum_{k=1}^Kp_{b \mid \theta_g, y} [ \hat b + L^{-\top}\{ c + S_k (b^*-c)\} ]$. 

\item Let $\theta^* = (b^{*\top}, \theta_g^\top)^\top$.
\begin{align} \label{eq:LSC_RVB_density_K}
\log p(y,\theta^*) = \log p(\theta_g) + \log \left( \frac{1}{K|L|}\sum_{k=1}^K p_{b,y\mid \theta_g} [ \hat b + L^{-\top} \{ c + S_k(b^* - c) \},y ] \right) 
\end{align}

\item
Let $\H$ and $\G$ be blockwise reflections groups acting on $\tilde{b}$ with reflection point $c$, with $\H \le \G$ ($\H$ is a subgroup of $\G$). Define $p_{\G}(z\mid \theta_g,y) = \frac{1}{|\G||L|}\sum_{S\in\G}p_{b\mid\theta_g,y} [\hat b + L^{-\top}\{c+S(z-c)\} ]$. Then $\KL\left\{\phi(\cdot\mid c,I) \, \| \, p_{\G}(\cdot\mid \theta_g,y)\right\} \leq \KL\left\{\phi(\cdot\mid c,I) \, \| \, p_{\H}(\cdot\mid \theta_g,y)\right\}$, with strict inequality unless $p_{\H}(\cdot\mid \theta_g,y)$ is $\T_{\G, c}$-invariant. 

\item $\KL\left\{\phi(b^*\mid c,I) \, \| \, p(b^*\mid \theta_g,y)\right\}
\leq \KL \{\phi(\tilde b\mid c,I) \, \| \, p(\tilde b\mid \theta_g,y) \}$. The inequality is strict unless $p_{\tilde b\mid\theta_g,y}(z)=p_{\tilde b\mid\theta_g,y}\{c+S_k(z-c)\}$ $\forall \, k=1,\dots,K$.
\end{enumerate}
\end{theorem}

Theorem \ref{thm: log_joint_density_skew_K} extends directly to the latent non-Gaussian model by additionally conditioning on the mixing variables $v$, and replacing $\hat{b}$ with the mode $\hat b_{\NG}$ of $p(b \mid \theta_g, v, y)$ and $L$ with the Cholesky factor $L_{\NG}$ of $Q_\NG + H(\hat{b}_\NG)$. The log joint density of the reparametrized latent non-Gaussian model with variables $\theta^*=(b^{*\top}, \theta_g^\top, v^\top, \eta)^\top$ is
\begin{align}\label{eq:LSC_RVB_density_K_nonGaussian}
\log p(y,\theta^*)
&=\log\left( \frac{1}{K|L_\NG|}\sum_{k=1}^K p_{b,y\mid v,\theta_g} [\hat b_{\NG} + L_{\NG}^{-\top}\{ c + S_k(b^* - c)\},y ] \right) 
\\ \nonumber
&\quad +\log p(\theta_g) +\log p(v\mid \eta) +\log p(\eta).
\end{align}

To illustrate the gains from skewness reduction, we simulate data from a Bernoulli GLMM with linear predictor $\eta_{ij} = -5 + 2x_{ij} + b_i$, where $x_{ij}\sim\Bernoulli(0.5)$ and $b_i\sim \N(0,3^2)$, for $i=1,\ldots,500$, $j=1,\ldots,3$. Figure \ref{Fig:skewness_correction}(a) shows that the 2-component skew-symmetric approximation of $p(b_1 \mid \theta_g, y)$ based on Theorem \ref{thm:K_optimal_skew_KL} substantially improves on the Gaussian approximation in Section \ref{sec:LGM}. Figure \ref{Fig:skewness_correction}(b) compares 10,000 draws of $\tilde{b}_1$ and $b_1^*$ from Procedure \ref{Algs:transformation_K}, taking $c=0$. The Q-Q plot aligns with Theorem \ref{thm: log_joint_density_skew_K}(iv), showing that skewness reduction yields $b_1^*$ much closer to a standard Gaussian than normalized-only $\tilde{b}_1$.

\begin{figure}[tb!]
\centering
\includegraphics[width=0.64\textwidth]{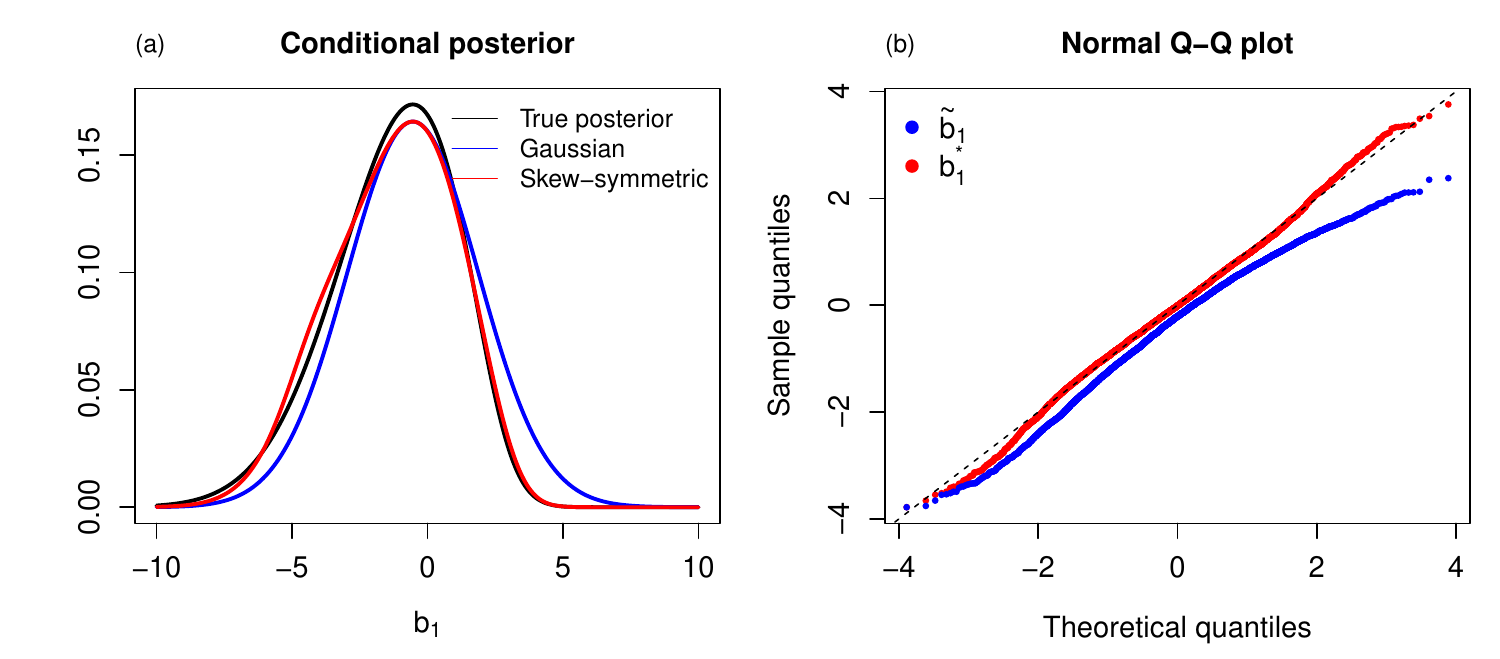}
\caption{\small Bernoulli GLMM. (a) Conditional posterior $p(b_1 \mid \theta_g, y)$ and its Gaussian and skew-symmetric approximations. (b) Q-Q plot of draws of $\tilde b_1$ and $b^*_1$.}\label{Fig:skewness_correction}
\end{figure}

A closely related approach is GLOSS-VA \citep{Kock2026} for latent Gaussian models with conditionally independent local variables. It builds on a conditionally structured Gaussian variational family, $q(\theta_g) \prod_{i} p(b_i \mid \theta_g)$ \citep{Tan2020B}, and introduces skewness via skewing functions for skew-symmetric densities, optimizing these jointly with the baseline variational parameters rather than as a post hoc correction \citep{Pozza2026}. While GLOSS-VA injects skewness and samples from skew-symmetric densities via the stochastic representation of \cite{Wang2004}, we instead remove skewness, using Proposition \ref{prop: K_skew2symm} to sample from the symmetrized density and Theorem \ref{thm: log_joint_density_skew_K} to marginalize the auxiliary sampling variables. A key distinction is that GLOSS-VA remains sensitive to parametrization, while the KNorm-RVB normalization can be viewed as a generalization of partial noncentering, which effectively eliminates the need to choose a parametrization.

\subsection{Choice of blockwise reflection group} \label{sec: choose m}
The blockwise reflection group $\G_m$ is defined by a partition $B_1, \dots, B_m$ of $\{1, \dots, N\}$ with $1 \leq m \leq N$, which induces $K=2^m$ blockwise reflections. While Theorem \ref{thm: log_joint_density_skew_K} favors large $m$ and hence large $K$ for stronger skewness reduction, this can be computationally prohibitive, because evaluating the reparametrized log joint density in \eqref{eq:LSC_RVB_density_K} and its gradients via automatic differentiation scales linearly in $K$, which is exponential in $m$.

A key exception arises when the conditional posterior factorizes, $p(b\mid\theta_g,y) = \prod_{i=1}^n p(b_i\mid\theta_g,y)$, e.g., in GLMMs with conditionally independent subject-level random effects. In this case, the Cholesky factor $L = \blockdiag(L_1, \dots, L_n)^\top$. Taking $m = n$ with subjectwise blocks $B_1, \dots, B_n$ yields $K=2^n$ blockwise reflections, but from Theorem \ref{thm: log_joint_density_skew_K}(i) and \eqref{eq:refl_gp}, 
\begin{align*}
&\log p(b^*\mid \theta_g,y) = \log \left( \frac{1}{|L|2^n}\sum\nolimits_{\varepsilon\in\{-1,1\}^n} p_{b \mid \theta_g, y} [\hat b+L^{-\top}\{c+S(\varepsilon)(b^*-c)\} ] \right) \\
&\quad= \log \left\{ \frac{1}{|L|2^n} \sum\nolimits_{\varepsilon\in\{-1,1\}^n} \left( \prod \nolimits_{i=1}^n p_{b_i \mid \theta_g, y} [\hat b_i+L_i^{-\top}\{c_i+\varepsilon_i(b_i^*-c_i)\} ]\right) \right\} \\
&\quad=\sum_{i=1}^n \log \left[ p_{b_i \mid \theta_g, y} (\hat b_i+L_i^{-\top}b_i^*) + p_{b_i \mid \theta_g, y} \{\hat b_i+L_i^{-\top}(2c_i-b_i^*)\} \right]
- \sum_{i=1}^n \log |L_i| - n \log 2.
\end{align*}
Thus, the computational cost grows linearly in $n$ rather than exponentially.

When $Q$ is not block diagonal, the conditional posterior does not factorize across subjects, so $m=n$ will require $K=2^n$ reflections and quickly becomes infeasible. We therefore partition $\{1, \dots, N\}$ into $m\ll n$ nearly equal blocks of sizes $s_1,\ldots,s_m$, keeping each subject’s observations within one block. Setting $m=1$ recovers the standard skew-symmetric representation, while $1 < m \ll n$ improves Gaussianization with manageable cost.

\section{Variational families and optimization algorithm} \label{sec:CSN}
This section specifies variational families for the local and global variables, as well as for the tail parameter and mixing variables in the latent non-Gaussian model. We then present the ELBO optimization algorithm for both model classes. Incorporating the skewness reduction from Section \ref{sec:K_skew_loc_cond_pos}, we work with reparametrized variables $\theta^*$. For the latent Gaussian model, $\theta^* = ({b^*}^\top, \theta_g^\top)^\top$ and $q_\lambda(\theta^*) = q(\theta_g) \prod_{i} q(b_i^*)$, whereas for the latent non-Gaussian model, $\theta^*=(b^{*\top}, \theta_g^\top, v^\top, \eta)^\top$ and $q_\lambda(\theta^*) = q(\eta) q(\theta_g) \prod_{i} q(b_i^*) \prod_j q(v_{ij})$. We use the same mean-field variational approximations as in Section \ref{sec:g_RVB} with $\tilde{b}_i$ replaced by $b^*_i$.

Since KNorm-RVB makes $\tilde{b}$ and $b^*$ approximately Gaussian, we take $q(\tilde{b}_i)$ and $q(b_i^*)$ as $\N(\mu_i, C_i C_i^\top)$. In contrast, the posteriors of $\theta_g$, $\eta$ and $v_{ij}$ may be skewed or heavy-tailed, motivating richer variational families. We map any constrained variables to $\mathbbm{R}$, e.g., $\ell_\eta = \log(\eta)$ and $\ell_{v_{ij}} = \log(v_{ij})$, and adopt the CSN subclass \citep{Tan2025} for $q(\theta_g)$, $q(\ell_\eta)$ and $q(\ell_{v_{ij}})$. This subclass captures skewness while remaining tractable, and admits a bounding line in each dimension. It reduces to the skew normal \citep[SN,][]{Azzalini1999} in one dimension. More flexible choices, such as skew-$t$ families or normalizing flows \citep{Papamakarios2021}, can likewise be paired with KNorm-RVB.

As the CSN subclass is only recently introduced in VI, we briefly summarize the key properties used here. For $x \in \mathbb{R}^d$, write $x = (x_i)$ and let $D_x = \diag(x)$. The CSN subclass is defined by an affine transformation of normalized independent  univariate SN random variables. Let $r_i \sim \SN(0,1,\gamma_i)$ independently for $i=1,\dots,d$, where $\gamma_i \in\mathbb{R}$ is the shape parameter, and define $r=(r_i)$ and $\gamma=(\gamma_i)$. Then $\E(r) = c_0 \delta$ and $\Cov(r) = D_\tau^2$, where $c_0=\sqrt{2/\pi}$, $\delta = (\delta_i)$ with $\delta_i = \gamma_i/\sqrt{1+\gamma_i^2}$ and $\tau = (\tau_i)$ with $\tau_i = \sqrt{1-c_0^2\delta_i^2}$. We normalize $r$ via $z = D_\tau^{-1}(r-c_0\delta)$ and define $\vartheta = \mu_\vartheta + C_\vartheta z$, so that $\E(\vartheta) = \mu_\vartheta$ and $\Cov(\vartheta) = C_\vartheta C_\vartheta^\top$ (with invertible $C_\vartheta$). By change of variables from $r$ to $\vartheta$, 
\begin{equation*}
\log q(\vartheta) = d\log(2) - \frac{d}{2}\log(2\pi) - \frac{r^\top r}{2} - \log |C_\vartheta| + \sum_{i=1}^d \{\log \Phi (\gamma_i r_i) + \log \tau_i\},
\end{equation*}
where $r=D_\tau C_\vartheta^{-1}(\vartheta - \mu_\vartheta) + c_0 \delta$ and $\Phi(\cdot)$ is the standard normal distribution function.

To use the reparametrization trick in ELBO optimization, write $\vartheta= \mu_\vartheta + C_\vartheta (D_\kappa w_2+D_\alpha \widetilde w_1)$, where $w_1, w_2 \overset{iid}{\sim} \N(0, I_d)$, $\widetilde w_1=|w_1|-c_0\bfone$, $\kappa = (\kappa_i)$ with $\kappa_i = \{1+(1-c_0^2)\gamma_i^2\}^{-1/2}$ and $\alpha=(\alpha_i)$ with $\alpha_i = \delta_i/\tau_i$. \cite{Tan2025} note that the CSN subclass ELBO has a stationary point at $\gamma=0$ which can hinder optimization, and propose a centered parametrization in terms of $\mu_\vartheta$, $C_\vartheta$ and $\alpha^3 = (\alpha_i^3)$ (Pearson's index of skewness) instead of $\gamma$. They also argue that restricting $C_\vartheta$ to a lower triangular Cholesky factor limits permissible rotations, and can reduce the ability of $q(\vartheta)$ to capture posterior asymmetries. Instead, they recommend an LU decomposition to preserve invertibility during stochastic optimization. Accordingly, we set $C_\vartheta = L_\vartheta U_\vartheta$, where $L_\vartheta$ is lower triangular and $U_\vartheta$ is  upper triangular with unit diagonal, and optimize $\{\mu_\vartheta, L_\vartheta, U_\vartheta, \alpha^3\}$ as the variational parameters of $q(\vartheta)$.

\begin{algorithm}[tb!]
    \caption{Stochastic optimization of ELBO in KNorm-RVB}
    \label{Algs:elbo_opti}
    \begin{algorithmic}[1]
    \Input Variational parameter $\lambda$, stepsize $\{\rho_s\}_{s=1}^T$, and blockwise reflection group $\G_m$
    \For{$s=1,\ldots,T$}
        \State $b_i^* = \mu_i + C_i z_i$ for $z_i \sim \N(0,I)$ $\forall\, i$ and set $c = (\mu_1^\top, \dots, \mu_n^\top)^\top$
        \State \parbox[t]{0.46\linewidth}{\raggedright\textbf{Latent Gaussian}}\hfill
               \parbox[t]{0.46\linewidth}{\raggedright\textbf{Latent non-Gaussian}}
        \State \parbox[t]{0.46\linewidth}{\raggedright Draw $\epsilon \sim \pi(\epsilon)$ and $\theta_g \gets t(\epsilon,\lambda)$}\hfill
               \parbox[t]{0.46\linewidth}{\raggedright Draw $\epsilon \sim \pi(\epsilon)$ and $(\theta_g, v, \eta) \gets t(\epsilon,\lambda)$}
        \State \parbox[t]{0.46\linewidth}{\raggedright$\theta^* \gets (b^{*\top},\theta_g^\top)^\top$}\hfill
               \parbox[t]{0.46\linewidth}{\raggedright$\theta^* \gets (b^{*\top}, \theta_g^\top, v^\top, \eta)^\top$}
        \State \parbox[t]{0.46\linewidth}{\raggedright Evaluate $\log p(y,\theta^*)$ using \eqref{eq:LSC_RVB_density_K}}\hfill
               \parbox[t]{0.46\linewidth}{\raggedright Evaluate $\log p(y,\theta^*)$ using \eqref{eq:LSC_RVB_density_K_nonGaussian}}
        \State Compute an ELBO estimate $\widehat{\L} \gets \log p(y,\theta^*)-\log q_{\lambda}(\theta^*)$
        \State Update $\lambda \gets\lambda + \rho_s \widehat{\nabla}_\lambda \L$, where $\widehat{\nabla}_\lambda \L \gets  \nabla_\lambda \theta^* \,  \{\nabla_{\theta^*}\log p(y, \theta^*) - \nabla_{\theta^*} \log q_\lambda (\theta^*) \} $
    \EndFor
    \end{algorithmic}
\end{algorithm}

Algorithm \ref{Algs:elbo_opti} summarizes stochastic optimization of the ELBO for KNorm-RVB in a general setting, where $q(\theta_g)$, $q(\ell_\eta)$ and $q(\ell_{v_{ij}})$ need not lie in the CSN subclass. We collect all variational parameters in $\lambda$ and form an unbiased estimate of $\nabla_\lambda \L$ using the reparametrization trick with sticking-the-landing (see section \ref{sec: RVB review}). The reflection point $c$ is set to the mean $(\mu_1^\top, \dots, \mu_n^\top)^\top$ of the Gaussian variational approximation for $b^*$, so $c$ is updated jointly with the other variational parameters in Algorithm \ref{Algs:elbo_opti}. Although $q(\theta^*)$ is defined in terms of $b^*$, posterior draws of original local variables $b$ can be recovered using Theorem \ref{thm:K_optimal_skew_KL} and Proposition \ref{prop: K_skew2symm}(i), and the sampling algorithm is given in Procedure \ref{Algs:forward_transformation_K} in the supplement.

\section{Exact recovery of mean and correlation matrix}\label{sec:theory}
Given a reflection point $c$ and the blockwise reflection group $\G_m$ ($m \geq 1$), KNorm-RVB (Procedure \ref{Algs:transformation_K}) normalizes local variables and reduces skewness, yielding $b^*$, which  is closer to $\N(c, I)$ than $\tilde{b}$ from RVB by Theorem \ref{thm: log_joint_density_skew_K}(iv). We next study the symmetries of $b^*$ induced by KNorm-RVB. Following \cite{Margossian2024}, we analyze VI robustness to misspecification (when the true posterior lies outside the variational family) in settings where $p(b^* \mid \theta_g, y)$ satisfies symmetry properties and $q(b^*)$ is a symmetry-preserving location-scale family. These results quantify KNorm-RVB's gains over RVB, inform the choice of $q(\theta^*)$, and guarantee recovery of the true posterior mean and correlation matrix under regularity conditions. We focus on the latent Gaussian model with $\theta^* = ({b^*}^\top, \theta_g^\top)$, and extensions to latent non-Gaussian models follow by additionally conditioning on the mixing variables.

We begin by defining the variational family and relevant symmetry notions. Throughout, we assume $q(\theta^*) = q(b^*)q(\theta_g)$, and let $\Q$ denote the variational family for $q(b^*)$. For the analysis, we assume $p(b^* \mid \theta_g, y)$ is differentiable, $\log p(b^* \mid \theta_g, y)$ and all $q \in \Q$ satisfy regularity conditions permitting differentiating under the integral sign, each $q \in \Q$ has finite moments of all orders, and $\| \nabla_{b^*} \log p(b^* \mid \theta_g, y) \|$ is bounded by a polynomial in $\|b^*\|$.

\begin{definition} \label{defn: location-scale family}
Let $q_0$ be a base pdf on $\mathbb{R}^N$. A location-scale family contains densities of the form $q_{\nu, S}(x) = q_0\left(  S^{-1/2} (x - \nu) \right) |S|^{-1/2}$ $\forall \, x\in \mathbb{R}^N$, for any location $\nu \in \mathbb{R}^N$ and positive definite scale matrix $S \in \mathbb{R}^{N \times N}$. 
\end{definition}

\begin{definition}
A function $f:\mathbb{R}^N \rightarrow \mathbb{R}$ is even-symmetric about $\nu \in \mathbb{R}^N$ if $f(\nu +x) = f(\nu-x)$ $\forall\, x \in \mathbb{R}^N$, and odd-symmetric about $\nu \in \mathbb{R}^N$ if $f(\nu +x) = - f(\nu-x)$ $\forall\, x \in \mathbb{R}^N$.
\end{definition}

\begin{definition}
A pdf $f$ on $\mathbb{R}^N$ is spherically symmetric if $f(x)$ depends on $x$ only through $\|x\|$. That is, if $\|x_1\| = \|x_2\|$, then $f(x_1) = f(x_2)$.
\end{definition}

\begin{definition} \label{defn: elliptical symmetry}
A pdf $f$ on $\mathbb{R}^N$ is elliptically symmetric about $\nu \in \mathbb{R}^N$ if $\exists$ a positive definite matrix $M \in \mathbb{R}^{N \times N}$ such that the pdf of $M^{-1/2}(x - \nu)$ is spherically symmetric for $x \sim f$. Note: $\Cor(x_i, x_j) = M_{ij}/\sqrt{M_{ii} M_{jj}}$, so $M$ determines the correlation matrix.
\end{definition}

Theorem \ref{thm:reflection_symmetry} highlights how KNorm-RVB’s skewness reduction induces $\T_{\G_m, c}$-invariance in $p(b^* \mid \theta_g, y)$. The resulting even-symmetry about $c$ forces the local posterior mean to equal $c$ and the local-global posterior cross-covariance to be exactly zero whenever these moments exist. This is stronger than RVB's normalization, which yields $\tilde{b}$ that is only approximately standard Gaussian. These results motivate the variational factorization $q(\theta^*) = q(b^*)q(\theta_g)$. and the setting of $c$ to be the local variational mean.

\begin{theorem}
\label{thm:reflection_symmetry}
The conditional posterior $p(b^*\mid \theta_g,y)$ is $\T_{\G_m,c}$-invariant and even-symmetric about $c$. This implies that $\E(b^*\mid\theta_g,y) = \E(b^* \mid y) = c$ and $\Cov(b^*,\theta_g\mid y)=0$ provided these moments exist.
\end{theorem}

\begin{theorem}[Exact recovery of mean]\label{thm:mean_recovery}
Let $q_{\nu,\lambda}(\theta^*) = q_{\lambda_{b^*}}(b^*-\nu) q_{\lambda_g}(\theta_g)$, where $q_{\lambda_{b^*}}$ is a pdf on $\mathbb R^N$ that is even-symmetric about the origin and $\nu \in \mathbb R^N$ is a location parameter. For any given $\lambda = (\lambda_{b^*}^\top, \lambda_g^\top)^\top$, $\KL\{q_{\nu,\lambda}(\theta^*)\|p(\theta^*\mid y)\}$ has a stationary point at $\nu=c$, which is unique if $\log p(b^*\mid \theta_g,y)$ is concave on $\mathbb R^N$ almost everywhere (a.e.) with respect to $q_{\lambda_g}$ and strictly concave on a nonempty open set with positive probability under $q_{\lambda_g}$.
\end{theorem}

Theorem \ref{thm:mean_recovery} shows that under a location variational family $\Q$ with even-symmetry, MFVI recovers the mean of $p(b^* \mid \theta_g, y)$ exactly. This supports using a Gaussian variational family for $b^*$. Correlation recovery requires stronger symmetry. Theorem \ref{thm:cor_recovery} shows that under a location-scale variational family $\Q$, MFVI recovers the covariance matrix up to a scale factor and therefore recovers the correlation matrix when $p(b^* \mid \theta_g, y)$ is elliptically symmetric about $c$. The  proofs of Theorems \ref{thm:mean_recovery} and \ref{thm:cor_recovery} follow Theorems 8 and 10 of \cite{Margossian2025a} and rely on their Lemma 13.

\begin{theorem}[Exact recovery of correlation matrix] \label{thm:cor_recovery}
Let $q_{\nu, S, \lambda_g}(\theta^*) = q_{\nu, S}(b^*) q_{\lambda_g}(\theta_g) $, where $\Q = \{q_{\nu, S}\}$ is the location-scale family (Definition \ref{defn: location-scale family}) with a spherically symmetric base pdf $q_0$. If $p(b^*\mid\theta_g,y)$ is elliptically symmetric about $c$ with fixed scale matrix $M$ a.e. with respect to $q_{\lambda_g}$, and $\log p(b^*\mid \theta_g,y)$ is concave on $\mathbb R^N$ a.e. with respect to $q_{\lambda_g}$ and strictly concave on a nonempty open set with positive probability under $q_{\lambda_g}$, then $\KL\{ q_{\nu, S, \lambda_g}(\theta^*) \| p(\theta^*\mid y)\}$ has a unique minimizer at $\nu = c$ and $S = \gamma^2 M$ for some $\gamma>0$ for any given $\lambda_g$. 
\end{theorem}

By Theorem \ref{thm:reflection_symmetry}, KNorm-RVB's skewness reduction enforces even-symmetry about $c$ in $p(b^* \mid \theta_g, y)$. This suffices for MFVI to  recover the mean exactly when $\Q$ is restricted to even-symmetric densities (Theorem \ref{thm:mean_recovery}). However, exact recovery of the correlation matrix requires the stronger property that $p(b^* \mid \theta_g, y)$ be elliptically symmetric about $c$ for a location-scale family $\Q$ (Theorem \ref{thm:cor_recovery}), which KNorm-RVB does not guarantee. Nevertheless, KNorm-RVB's normalization and skewness reduction tend to steer the distribution of $b^*$ toward $\N(c, I)$, and thus toward elliptical symmetry. Theorem \ref{thm: log_joint_density_skew_K}(iii) further suggests that this Gaussianizing effect strengthens as the blockwise reflection group grows. Thus, when $p(b^* \mid \theta_g, y)$ is close to $\N(c,I)$, a Gaussian approximation centered at $c$ recovers the true center and can approximate the correlation structure well. Remaining discrepancy may be due to residual scale mismatch, tail behavior, or higher-order departures from Gaussianity.

\section{Applications}\label{sec:app}
We evaluate Algorithm \ref{Algs:elbo_opti} on GLMMs, mixed multinomial logit (MMNL) models, SAR models, and stochastic volatility models, covering both Gaussian and non-Gaussian latent fields. We compare KNorm-RVB (normalization plus skewness reduction) with RVB (normalization only), INLA for latent Gaussian models, structured variational inference \citep[SVI,][]{Cabral2024} for latent non-Gaussian models, and GLOSS-VA for latent Gaussian models with conditionally independent local variables.

For RVB and KNorm-RVB, the suffix -G or -CSN indicates whether the variational family for non-local variables is Gaussian or from the CSN subclass. For KNorm-RVB, the blockwise reflection group size $K$ is shown in parentheses. For GLMMs and MMNL models, the conditional local posterior factorizes, so $K=2^n$ unless stated otherwise.

We run RVB and KNorm-RVB for 50,000 iterations, using Adam \citep{Kingma2015} with default settings. For a warm start, CSN variational parameters are initialized from a fitted Gaussian variational approximation. As a gold standard, we run MCMC using the No-U-Turn Sampler in RStan using two parallel chains of 50,000 iterations each, discarding the first half of each chain as burn-in. All experiments are run on a 16GB Apple M1 machine using R and Python 3.14.3 with JAX 0.92. 

Multivariate accuracy relative to MCMC is assessed by the maximum mean discrepancy (MMD) criterion \citep{Zhou2023}. We define $M^*=-\log\{\max(\mathrm{MMD}_u^2, 0)+10^{-5}\}$, where
\begin{equation*}
\begin{aligned}
\mathrm{MMD}_u^2
= \frac{1}{m(m-1)} \sum_{i\neq j}^{m}
\Bigl[
k\!\left(\mathbf{x}_v^{(i)},\mathbf{x}_v^{(j)}\right)
+ k\!\left(\mathbf{x}_{g}^{(i)},\mathbf{x}_{g}^{(j)}\right)
- k\!\left(\mathbf{x}_v^{(i)},\mathbf{x}_{g}^{(j)}\right)
- k\!\left(\mathbf{x}_v^{(j)},\mathbf{x}_{g}^{(i)}\right)
\Bigr]
\end{aligned}
\end{equation*}
is an unbiased estimate of the squared MMD. Here $\mathbf{x}_v^{(i)}$ and $\mathbf{x}_g^{(i)}$ denote draws from the variational approximation and MCMC respectively, and $k(\cdot, \cdot)$ is the radial basis function kernel. We set $m=1000$ and compute $M^*$ over 50 repetitions, with larger values indicating closer agreement with the target. We do not report $M^*$ for SVI because INLA would need to refit the model and generate $(\theta_g, b)$ draws for each sampled $v$, which is computationally intensive. Alternatively, we report two marginal diagnostics, the absolute difference in posterior means and ratio of posterior standard deviations, both relative to MCMC. The mean Monte Carlo ELBO estimate, $\bar{\L}$, computed from 1000 simulations, also serves as a diagnostic for how well different variational methods approximate the true posterior.

\subsection{Generalized linear mixed model}
Let $y_i = (y_{i1},\dots,y_{in_i})^\top$ denote the observations for subject $i$, $i=1, \dots, n$. In a GLMM, $y_{ij}$ follows an exponential family distribution, and $g(\E(y_{ij})) = \eta_{ij} =X_{ij}^\top \beta + Z_{ij}^\top b_i$, where $g(\cdot)$ is a link function, $X_{ij} \in \mathbb{R}^p$ and $Z_{ij} \in \mathbbm{R}^r$ are covariates, $\beta \in \mathbb{R}^p$ are fixed effects and $b_i \overset{iid}{\sim} \N(0, B^{-1})$ are random effects. Stacking $b = (b_1^\top, \dots, b_n^\top)^\top$ gives $b \sim \N(0,Q^{-1})$ with $Q = \blockdiag(B, \dots, B)$. For unconstrained optimization, we use the Cholesky factorization $B = L_B L_B^\top$, where $L_B$ is lower triangular with positive diagonal. Let $L_B^*$ satisfy $L_{B, ii}^* = \log L_{B, ii}$ and $L_{B, ij}^* = L_{B, ij}$ for $i\neq j$, and define $\zeta = \vech(L_B^*)$. The global variables are $\theta_g=(\beta^\top,\zeta^\top)^\top$, with prior $\N(0,100I)$. For the latent non-Gaussian extension, $Q = D^\top D$ with $D = \blockdiag(L_B^\top, \dots, L_B^\top)$.

The polypharmacy data set \citep{Hosmer2013} follows 500 subjects annually for 7 years, yielding 7 binary drug-use responses per subject. Covariates include gender (male = 1, female = 0), race (non-White = 1, White = 0), log(age/10), outpatient mental health visit indicators (MHV1 = 1 for 1--5 visits, MHV2 = 1 for 6--14 visits, MHV3 = 1 for $\geq$ 15 visits; 0 otherwise), and an indicator for any inpatient mental health visits (yes=1, no=0). We fit a logistic random intercept model and consider KNorm-RVB with $K \in \{2, 2^n\}$, corresponding to the standard and subjectwise skew-symmetric representations.

\begin{figure}[tb!]
\centering
\includegraphics[width=\textwidth]{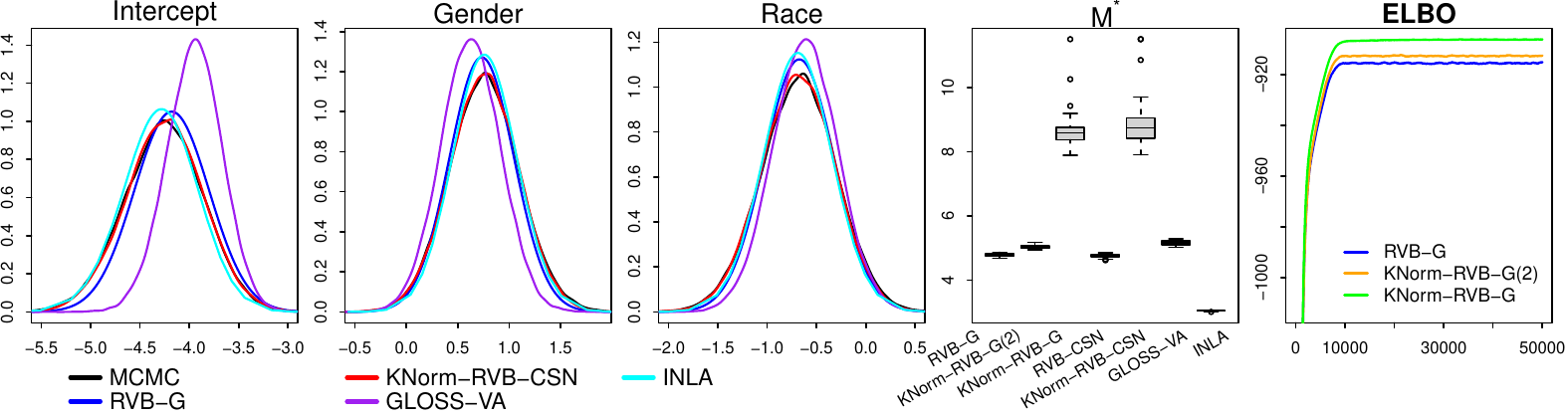}
\caption{Polypharmacy latent Gaussian model. Marginal posteriors for selected global variables, boxplot for $M^*$ and ELBO averaged over every 1000 iterations.} \label{Fig: poly_lgm}
\end{figure}

Figure \ref{Fig: poly_lgm} summarizes results for the latent Gaussian model. In the first three marginal posterior plots, KNorm-RVB-CSN matches the MCMC benchmark most closely, improving on INLA, RVB-G and GLOSS-VA. The boxplots of $M^*$ show that KNorm-RVB improves slightly on RVB when $K=2$ and substantially when $K=2^n$. Replacing the Gaussian variational approximation for the global variables with the CSN subclass yields a further modest gain in $M^*$. Overall, KNorm-RVB-G and KNorm-RVB-CSN achieve markedly higher $M^*$ than GLOSS-VA and INLA. The accuracy gains from RVB to KNorm-RVB by increasing $K$ from 1 to 2 to $2^n$ are also reflected in the ELBO trace plots for the Gaussian variant.

For the latent non-Gaussian model, the first two panels of Figure \ref{Fig: poly_logV_lgm} show that the normalization used in RVB and KNorm-RVB, which renders the transformed local variables approximately independent of the mixing variables $v$, improves estimation of the marginal posterior mean, and especially the variance, of $\log(v)$ relative to SVI. Increasing the blockwise reflection group size from $K$ from 2 to $2^n$ yields further accuracy gains, which are also evident in the $M^*$ and ELBO plots. Overall, KNorm-RVB-CSN attains the highest $M^*$, with KNorm-RVB-G a close second, highlighting the benefit of combining normalization with skewness reduction.

\begin{figure}[tb!]
\centering
\includegraphics[width=\textwidth]{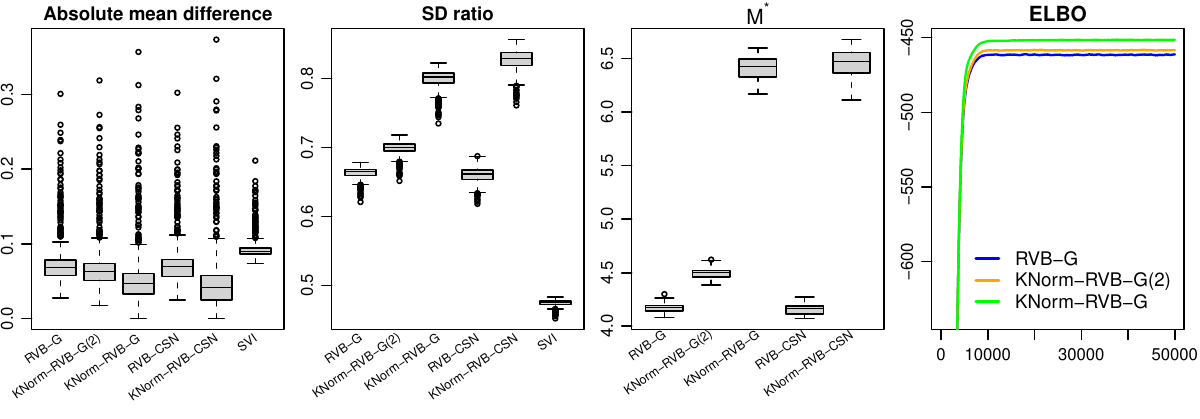}
\caption{Polypharmacy latent non-Gaussian model. Boxplots of absolute mean difference, standard deviation ratios for $\log(v)$, and ELBO averaged over every 1000 iterations.} \label{Fig: poly_logV_lgm}
\end{figure}

\subsection{Mixed multinomial logit model}
Let $y_{it}\in\{1,\dots,J\}$ denote the alternative chosen by subject $i$ on choice occasion $t$ for $i=1, \dots, n$ and $t=1, \dots, T$. In the MMNL model, $\P(y_{it}= j \mid b_i,\beta)=\exp(U_{itj})/\sum_{k=1}^{J}\exp(U_{itj})$, with $U_{itj}=X_{itj}^\top\beta+Z_{itj}^\top b_i$, where $\beta$ are fixed effects and $b_i$ are subject-specific random effects. We consider the latent Gaussian model with $b_i\overset{iid}{\sim}\N(0,B^{-1})$ with the prior by \cite{Huang2013}, where $\beta\sim\N(0,v_0 I)$, $B^{-1}\mid a_1,\dots,a_r\sim \mathrm{IW}\{\nu+r-1,2\nu\diag(a_1,\dots,a_r)\}$ and $a_\ell\sim\mathrm{Gamma}(1/2,1/A^2)$ for $\ell=1,\dots,r$. Here, $\mathrm{IW}(\cdot, \cdot)$ denotes the inverse Wishart distribution, and we set $v_0=10^6$, $\nu=2$ and $A=10^3$. For unconstrained optimization, write $B=L_B L_B^\top$ and define $\zeta=\vech(L_B^*)$, where $L^*_{B,ii}=\log L_{B,ii}$ and $L^*_{B,ij}=L_{B,ij}$ for $i\neq j$. The global variables are $\theta_g=(\beta^\top,\zeta^\top,\log a_1,\dots,\log a_r)^\top$.

We analyze the Electricity data from the \texttt{mlogit} R package \citep[see e.g.,][]{Tan2017b}, comprising $n=361$ residential customers and $T=4308$ choice occasions in total. For each occasion $t$, the response $y_{it}$ records the selected supplier among $J=4$ alternatives. Supplier attributes include fixed price, contract length, indicators for local and well-known suppliers, and time-of-day and seasonal rates. We include all attributes as fixed effects in the MMNL model, and allow a random effect for fixed price to capture heterogeneous price sensitivity. Since the conditional local posterior factorizes by respondent, we adopt a subjectwise reflection group with $K=2^n$.

\begin{figure}[tb!]
\centering
\includegraphics[width=\textwidth]{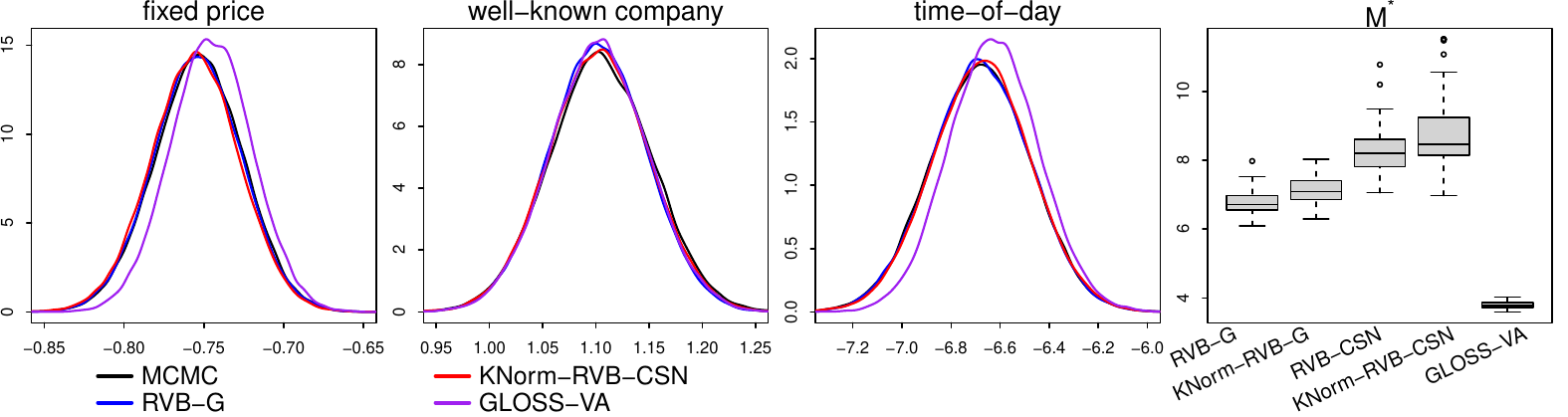}
\caption{\small Latent Gaussian model for electricity data. Marginal posteriors for selected global variables, boxplot for $M^*$.} \label{Fig: ele_lgm}
\end{figure}

For the latent Gaussian model, we compare RVB-G, RVB-CSN, KNorm-RVB-G and KNorm-RVB-CSN against MCMC, and include GLOSS-VA as an additional benchmark. We omit INLA because applying it to the MMNL likelihood will require a Poisson-trick representation \citep{Baker1994}, which expands the data into alternative-level pseudo-observations and is substantially slower than the variational methods. Figure \ref{Fig: ele_lgm} shows that RVB-G and KNorm-RVB-CSN tack the MCMC marginal posteriors more closely than GLOSS-VA. The $M^*$ results are consistent with these plots, showing that RVB and KNorm-RVB outperform GLOSS-VA, and KNorm-RVB improves upon RVB for both the Gaussian and CSN subclass variants, reflecting the benefit of the additional skewness reduction step.

Table \ref{tab:elbo} shows that, for the latent non-Gaussian model, KNorm-RVB achieves higher mean Monte Carlo ELBO $\bar{\L}$ than RVB for both the Gaussian and CSN subclass variants. We do not report accuracy relative to MCMC because obtaining a reliable MCMC benchmark for this model is substantially more challenging and computationally intensive.

\begin{table}[tb!]
\centering
\begin{tabular}{lrrrr}
\hline
Data & RVB-G    & KNorm-RVB-G & RVB-CSN  & KNorm-RVB-CSN \\ \hline
Electricity & $-3951.02$ & $-3949.57$ & $-3950.91$ & $-3949.51$ \\ 
SIDS & 1044.24 & 1044.84         & 1048.63 & 1049.34           \\ 
GBP  & 1109.87 & 1110.08         & 1109.98 & 1110.17           \\ \hline
\end{tabular}
\caption{Estimated ELBO, averaged over 1000 Monte Carlo simulations ($\bar{\L}$) for the latent non-Gaussian models. KNorm-RVB uses the largest $K$ considered in each example.} 
\label{tab:elbo}
\end{table}

\subsection{Spatial autoregressive model}
Here we analyze the sudden infant death syndrome (SIDS) data \citep{Cressie2015}, comprising the counts $\{y_i\}$ of SIDS across the $n=100$ counties of North Carolina (1974--1978). We model $y_i\mid \eta_i \sim \Poisson(E_i\eta_i)$, where $E_i$ is the expected count (overall SIDS rate times births in county $i$) and $\log(\eta_i) = X_i^\top \beta + b_i$ for $i=1,\dots,n$. The covariates $X_i$ include an intercept and the proportion of non-White births (standardized to have mean 0 and variance 1). Here $\beta$ are fixed effects and $b_i$ is a latent spatial effect. For an autoregressive prior \citep{Kissling2008,Ver2018} on $b$, let $W$ be the row-standardized adjacency matrix (each entry is divided by its row sum) and define $D = I_n-\rho W$ with $\rho\in(-1,1)$. In the latent Gaussian case, $D b = Z$ with $Z \sim \N(0, I)$, implying $b\sim\N(0, Q^{-1})$ where $Q = D^\top D$. We reparametrize $\rho=\tanh(\phi)$ and assign the prior $\theta_g = (\beta^\top,\phi)^\top\sim\N(0,100I)$. The latent non-Gaussian extension replaces the Gaussian driving noise $Z$ by $\epsilon$ whose components are independent NIG as in Section \ref{sec:lngm}.

\begin{figure}[tb!]
\centering
\includegraphics[width=\textwidth]{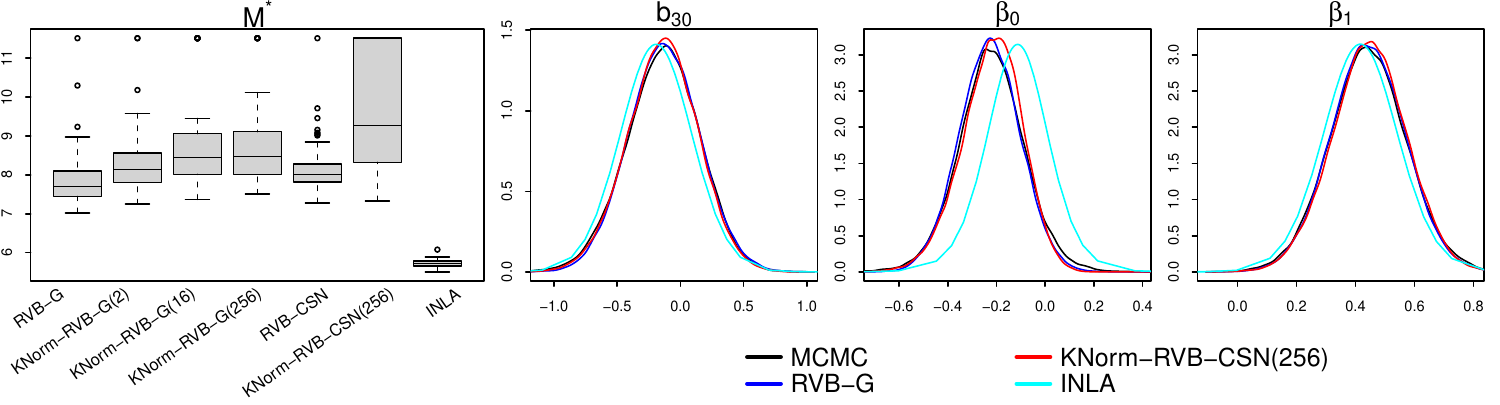}
\caption{SIDS latent Gaussian model: Boxplots of $M^*$ and marginal posteriors for $b_{30}$, $\beta_{0}$ and $\beta_{1}$.} \label{Fig: sids_lgm}
\end{figure}

Figure \ref{Fig: sids_lgm} shows that KNorm-RVB-G (2) achieves a higher $M^*$ than RVB-G, and that $M^*$ increases further as the blockwise reflection group size grows from $K=2$ to $K=256$ for the latent Gaussian model. This suggests improved variational approximation under richer blockwise symmetrization, consistent with Theorem \ref{thm: log_joint_density_skew_K}(iii). Moreover, all RVB and KNorm-RVB methods outperform INLA in terms of $M^*$, indicating better multivariate accuracy. The marginal posterior plots of $b_{30}$, $\beta_{0}$ and $\beta_{1}$ also show that RVB-G and KNorm-RVB-CSN (256) are closer to MCMC than INLA.

\begin{figure}[b!]
\centering
\includegraphics[width=\textwidth]{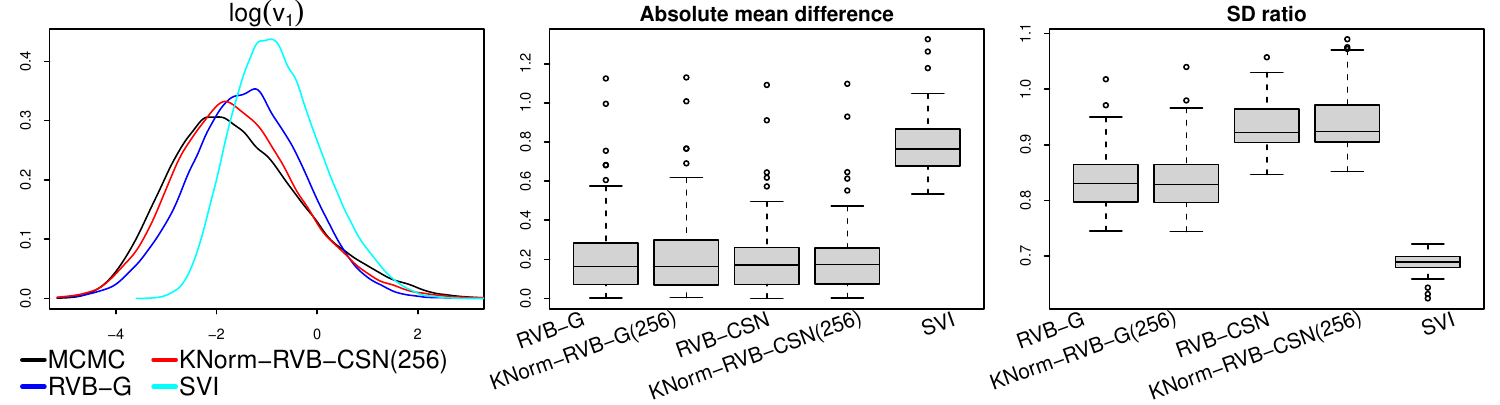}
\caption{\small SIDS latent non-Gaussian model: Marginal distributions of the mixing variable $\log(v_{1})$, together with boxplots of the absolute mean differences, standard deviation ratios for $\log(v)$.} \label{Fig: sids_lngm}
\end{figure}

For the latent non-Gaussian model, we report only the best-performing KNorm-RVB variant with $K=256$. The first panel in Figure \ref{Fig: sids_lngm} shows that RVB-G and KNorm-RVB-CSN (256) track the marginal posterior $\log(v_1)$ more closely than INLA, and that the CSN variant improves on the Gaussian approximation by capturing posterior skewness. The absolute mean differences and standard deviation ratios for $\log(v)$ indicate that RVB and KNorm RVB outperform SVI, producing posterior means and standard deviations closer to MCMC. Table \ref{tab:elbo} also shows that KNorm-RVB achieves higher $\bar{\L}$ values than RVB for both the Gaussian and CSN subclass, suggesting that the addition of the skewness reduction step improves the overall  variational fit.

\subsection{Stochastic volatility model}
In the stochastic volatility model, which is widely used for financial time series, the observations are conditionally Gaussian with zero mean and time-varying variance: $y_t \mid b_t \sim \N \left( 0,\exp(\lambda+\sigma b_t) \right)$ for $t = 1, \dots, n$, where $\lambda\in\mathbb{R}$ and $\sigma>0$. The latent log-volatility $\{b_t\}$ follows a stationary AR(1) process, with $b_1\sim\N (0, (1-\phi^2)^{-1} )$,
\[
b_t\mid b_{t-1}\sim\N(\phi\, b_{t-1},1),\quad t=2,\dots,n,
\]
and persistence parameter $0 < \phi <1 $. Hence $b = (b_1, \dots, b_n)^\top \sim \N(0, Q^{-1})$, where $Q$ is symmetric tridiagonal, with diagonal $(1, 1+\phi^2, \dots, 1+\phi^2, 1)^\top$ and off-diagonal entries $-\phi$. For unconstrained optimization, we reparametrize $\sigma=\log(1+e^\alpha)$ and $\phi=\{1+\exp(-\psi)\}^{-1}$, and collect the global variables as $\theta_g = (\alpha,\lambda,\psi)^\top$, with prior $\theta_g\sim\N(0,10I)$. For the latent non-Gaussian extension, $Q = D^\top D$ where $D$ is lower triangular, with diagonal $(\sqrt{1-\phi^2}, 1, \dots, 1)^\top$, first lower diagonal $(-\phi, \dots, -\phi)^\top$, and all remaining entries zero.

We consider the \texttt{Garch} dataset from the R package \texttt{Ecdat}, containing daily USD/GBP exchange rates from 1 October 1981 to 15 March 1984. For exchange-rate series $\{r_t\}$, define $y_t
=100\left\{ \log (r_t/r_{t-1}) - \frac{1}{n} \sum_{i=1}^n \log (r_i / r_{i-1}) \right\}$ for $t=1,\dots,n$, where $n=619$.

\begin{figure}[tb!]
\centering
\includegraphics[width=\textwidth]{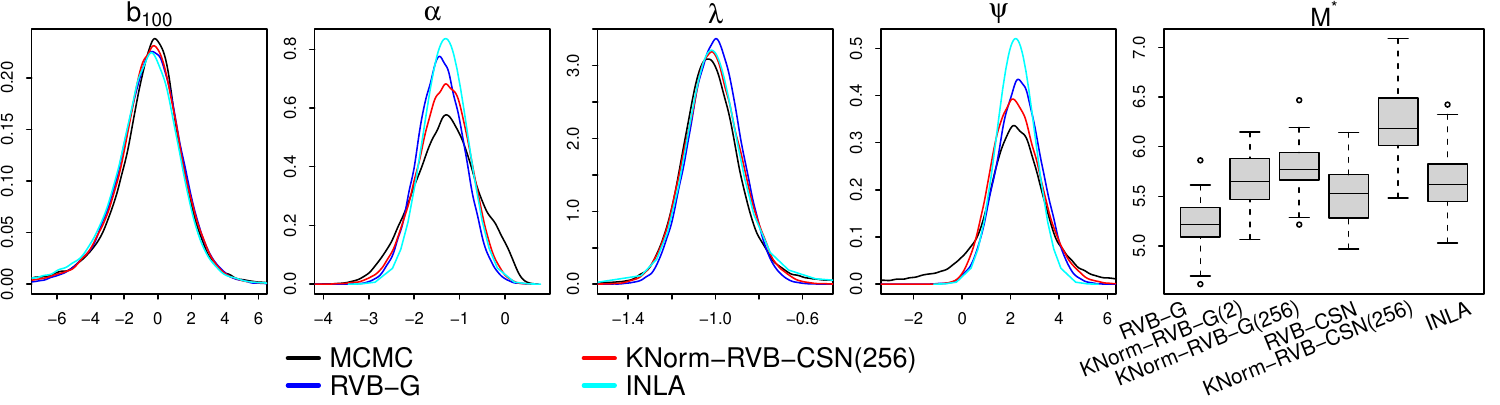}
\caption{GBP latent Gaussian model. Marginal posteriors for $b_{100}$ and $\theta_g$, and boxplots of $M^*$.} \label{Fig: gbp_lgm}
\end{figure}

For the latent Gaussian model, we consider KNorm-RVB with $K=2$ and $K=256$. Figure \ref{Fig: gbp_lgm} shows that KNorm-RVB-CSN(256) returns marginal posterior estimates for $b_{100}$ and the global parameters that are closer to MCMC than INLA and RVB-G. The $M^*$ results further indicate that KNorm-RVB improves upon RVB for both the Gaussian and CSN subclass. In particular, KNorm-RVB-CSN(256) attains a substantially higher $M^*$ than other competing methods, demonstrating the benefits of applying skewness reduction after normalizing the local variables, and highlighting the added flexibility of the CSN subclass.

\begin{figure}[b!]
\centering
\includegraphics[width=\textwidth]{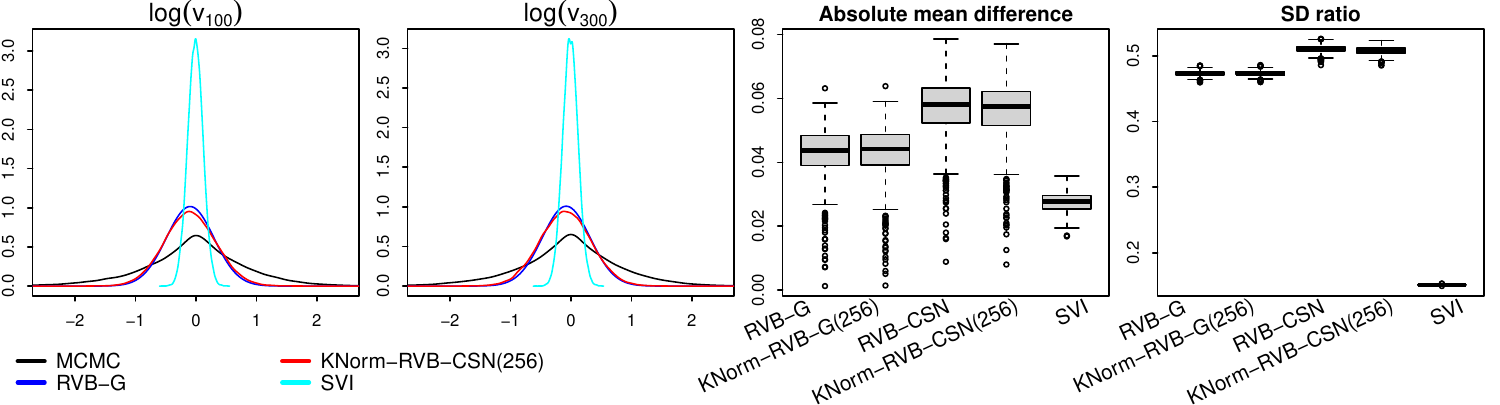}
\caption{GBP latent non-Gaussian model: Marginal distributions of $\log(v_{100})$ and $\log(v_{300})$, and boxplots of absolute mean differences and standard deviation ratios for $\log(v)$.} \label{Fig: gbp_lngm}
\end{figure}

For the latent non-Gaussian model, the marginal posterior plots for $\log(v_{100})$ and $\log(v_{300})$ in Figure \ref{Fig: gbp_lngm} show that RVB and KNorm-RVB markedly outperform SVI, which severely underestimates the posterior variance. For $\log(v)$, all methods yield small absolute mean differences, with SVI performing slightly better on this metric. However, RVB and KNorm-RVB deliver more accurate standard deviation ratios. Table \ref{tab:elbo} further shows that KNorm-RVB achieves higher ELBO values than RVB for both the Gaussian and CSN subclass, with the best result achieved by KNorm-RVB-CSN (256).

\section{Conclusion}\label{sec:conclusion}
This article develops KNorm-RVB, a reparametrization method that improves MFVI for hierarchical models. It generalizes RVB beyond GLMMs to latent Gaussian models with sparse local dependence, and latent non-Gaussian models with heavy-tailed noise. KNorm-RVB combines curvature-based normalization which makes local conditional posteriors approximately standard Gaussian, with skewness reduction via a novel $K$-component skew-symmetric representation. By enabling simulation from a blockwise symmetrized density, we derive an analytic reparametrized density, and show that Gaussianization strengthens as the blockwise reflection group grows. Our theory quantifies gains over RVB by proving that transformed local variables are centered at the reflection point and have zero posterior covariance with global variables. We also show that, under symmetry conditions on the variational family and conditional local posterior, MFVI recovers the exact local mean and correlation matrix. Although KNorm-RVB enforces even-symmetry rather than full elliptical symmetry, these results motivate normalization and symmetrization before MFVI, Gaussian approximations for transformed local variables, and setting the reflection point to the variational mean. Flexible CSN variational families are used for global and mixing variables, where skewness may remain. Across a range of models, KNorm-RVB improves posterior approximation accuracy over existing methods.

Our results suggest several promising directions for future work. The current reflection groups for skewness reduction exclude rotations, and incorporating rotation groups \citep[e.g.][]{Jupp2016} may further improve Gaussianization. Rather than selecting blockwise reflection partitions based on model structure and empirical performance, an adaptive rule driven by estimated local skewness may also improve efficiency. To scale KNorm-RVB to large datasets, the Taylor expansion centers used in Gaussian approximation of local conditional posteriors can be treated as auxiliary variational parameters which are updated infrequently without optimization, as in \citet{Zhang2026}. Since KNorm-RVB  primarily targets skewness, combining it with methods that address residual tail behavior may yield further gains. The $K$-component skew-symmetric representation may also be useful for constructing variational approximations directly. Finally, establishing convergence guarantees for the stochastic optimization algorithm in the spirit of \cite{Ko2024} remains an important open problem.

\section{Disclosure statement}\label{disclosure-statement}
There are no competing interests to declare.

\phantomsection\label{supplementary-material}
\bigskip
\begin{center}
{\large\bf SUPPLEMENTARY MATERIAL}
\end{center}

\begin{description}
\item[Technical supplement:]
The proofs of all propositions and theorems are provided in the supplementary material.
\end{description}


\begingroup
\spacingset{1}
\bibliography{ref}
\endgroup

\newpage

\noindent
\begin{center}
{ \bf \Large Supplementary Material}
\end{center}

\setcounter{section}{0} \renewcommand{\thesection}{S\arabic{section}}
\setcounter{figure}{0} \renewcommand{\thefigure}{S\arabic{figure}}
\setcounter{table}{0} \renewcommand{\thetable}{S\arabic{table}}
\setcounter{equation}{0} \renewcommand{\theequation}{S\arabic{equation}}
\setcounter{lemma}{0} \renewcommand{\thelemma}{S\arabic{lemma}}
\setcounter{theorem}{0} \renewcommand{\thetheorem}{S\arabic{theorem}}
\setcounter{proposition}{0} \renewcommand{\theproposition}{S\arabic{proposition}}
\setcounter{algorithm}{0} \renewcommand{\thealgorithm}{S\arabic{algorithm}}

\section{Proofs for Section \ref{sec:skewness_correction}}
In this section, we present the proofs of Propositions \ref{prop:K_skew_construction}, \ref{prop:K_skew_representation} and \ref{prop: K_skew2symm} for the construction, representation and stochastic representations of the $K$-component skew-symmetric density. We also provide the proofs of Theorems \ref{thm:K_optimal_skew_KL} and \ref{thm: log_joint_density_skew_K} on the optimal skewing function and effects of skewness reduction respectively.

\subsection{Proposition \ref{prop:K_skew_construction}: $K$-component skew-symmetric construction}
\begin{proof}
Let $g(\theta) = K f_{K,\hat\theta}(\theta) w_{K,\hat\theta}(\theta)$. It is straightforward to see that $g(\theta) \geq 0$, and it remains to show that $\int g(\theta) \, d\theta = 1$. As each $T_{k, \hat{\theta}}$ is a reflection (affine transformation) with absolute value of the Jacobian determinant ($|S_k|$) equal to one, by change of variables,
\begin{align*}
\int g(\theta) \, d\theta &= K \int f_{K,\hat\theta}(\theta) w_{K,\hat\theta}(\theta) \, d\theta\\
&= \sum_{k=1}^K \int f_{K,\hat\theta}\{T_{k, \hat{\theta}}(\theta)\}
w_{K,\hat\theta}\{T_{k, \hat{\theta}}(\theta)\} \, d\theta  \\
&= \int f_{K,\hat\theta}(\theta)
\sum_{k=1}^K w_{K,\hat\theta}\{T_{k, \hat{\theta}}(\theta)\} \, d\theta \\
&= \int f_{K,\hat\theta}(\theta) \, d\theta=1.
\end{align*}
The third equality uses the $\T_{\G_m, \hat{\theta}}$-invariance of $f_{K,\hat\theta}$, where $f_{K,\hat\theta} \{T_{k, \hat{\theta}}(\theta)\} = f_{K,\hat\theta}(\theta)$ $\forall$ $\theta \in \mathbbm{R}^d$ and $k=1,\dots,K$. The last equality uses the property of $w_{K,\hat\theta}$ as a $K$-component skewing function, where $\sum_{k=1}^K w_{K,\hat\theta}\{T_{k, \hat{\theta}}(\theta)\}=1$. Hence $g(\theta)$ is a pdf.
\end{proof}

\subsection{Proposition \ref{prop:K_skew_representation}: $K$-component skew-symmetric representation}
\begin{proof}
First, we show that $f_{K,\hat\theta} = \frac{1}{K}\sum_{k=1}^K g \{T_{k, \hat{\theta}}(\theta)\}$ is a $\T_{\G_m, \hat{\theta}}$-invariant pdf. As each $T_{k, \hat{\theta}}$ is a reflection (affine transformation) with absolute value of the Jacobian determinant ($|S_k|$) equal to one, by change of variables,
\begin{align*}
\int f_{K,\hat\theta} (\theta) \, d\theta 
&= \frac{1}{K}\sum_{k=1}^K \int g \{T_{k, \hat{\theta}}(\theta)\} d\theta 
= \frac{1}{K}\sum_{k=1}^K \int g(\theta)d\theta 
= 1.
\end{align*}
Hence $f_{K,\hat\theta}$ is a pdf. Note that $T_{k, \hat{\theta}}(T_{j, \hat{\theta}}(\theta)) = T_{k, \hat{\theta}}\{\hat{\theta} + S_j(\theta - \hat{\theta})\} = \hat{\theta} + S_k S_j(\theta - \hat{\theta})$ and $S_k S_j \in \G_m$ since $\G_m$ is a finite abelian group closed under matrix multiplication. Hence $T_{k, \hat{\theta}} \circ T_{j, \hat{\theta}}$ is a reflection induced by $\G_m$ about $\hat{\theta}$ for any $j,k=1, \dots, K$. Thus for any fixed $j=1, \dots, K$, $\{T_{k, \hat{\theta}} \circ T_{j, \hat{\theta}} \mid k=1,\dots,K\} = \{T_{k, \hat{\theta}} \mid k=1,\dots,K\}$. It follows that
\begin{align*}
f_{K,\hat\theta}\{T_{j, \hat{\theta}}(\theta)\}
&= \frac{1}{K}\sum_{k=1}^Kg\{T_{k, \hat{\theta}}(T_{j, \hat{\theta}}(\theta))\} \\
&= \frac{1}{K}\sum_{k=1}^Kg\{T_{k, \hat{\theta}}(\theta)\}=
f_{K,\hat\theta}(\theta)
\end{align*}
for any $j=1, \dots, K$ and $\theta \in \mathbbm{R}^d$. Thus $f_{K,\hat\theta}$ is a $\T_{\G_m, \hat{\theta}}$-invariant pdf.

Next, we show that $w_{K,\hat\theta}(\theta) = {g(\theta)}/{\sum_{k=1}^K g\{T_{k, \hat{\theta}}(\theta)\}}$ is a $K$-component skewing function, with $w_{K,\hat\theta}(\theta)=1/K$ if $\sum_{k=1}^K g\{T_{k, \hat{\theta}}(\theta)\}=0$. First, if $\sum_{k=1}^K g\{T_{k, \hat{\theta}}(\theta)\} > 0$,
\begin{align*}
\sum_{j=1}^K w_{K,\hat\theta}(T_{j, \hat{\theta}}(\theta)) 
&= \sum_{j=1}^K \frac{g(T_{j, \hat{\theta}}(\theta))}{\sum_{k=1}^K g(T_{k, \hat{\theta}}(T_{j, \hat{\theta}}(\theta)))}  \\
&= \sum_{j=1}^K \frac{g(T_{j, \hat{\theta}}(\theta))}{\sum_{k=1}^K g(T_{k, \hat{\theta}}(\theta))} \\
&=  \frac{\sum_{j=1}^K g(T_{j, \hat{\theta}}(\theta))}{\sum_{k=1}^K g(T_{k, \hat{\theta}}(\theta))} = 1.
\end{align*}
If $\sum_{k=1}^K g\{T_{k, \hat{\theta}}(\theta)\}=0$, then $\sum_{k=1}^K g\{T_{k, \hat{\theta}}(T_{j, \hat{\theta}}(\theta))\} = \sum_{k=1}^K g\{T_{k, \hat{\theta}}(\theta)\} = 0$ for any $j=1, \dots, K$. Hence $\sum_{j=1}^K w_{K,\hat\theta}(T_{j, \hat{\theta}}(\theta)) = \sum_{j=1}^K 1/K = 1$ for any $j=1, \dots, K$. Therefore, $w_{K,\hat\theta}(\theta)$ is a $K$-component skewing function.

For any pdf $g: \mathbb R^d \rightarrow \mathbb R_+$, if  $\sum_{k=1}^K g\{T_{k, \hat{\theta}}(\theta)\} > 0$, we can write 
\[
g(\theta) = \left[\frac{1}{K}\sum_{k=1}^Kg\{T_{k, \hat{\theta}}(\theta)\}\right]
\frac{K g(\theta)}{\sum_{k=1}^Kg\{T_{k, \hat{\theta}}(\theta)\}}
= Kf_{K,\hat\theta}(\theta)w_{K,\hat\theta}(\theta),
\]
If $\sum_{k=1}^K g\{T_{k, \hat{\theta}}(\theta)\} = 0$, then $w_{K,\hat\theta}(\theta)=1/K$ while $f_{K,\hat\theta} = \frac{1}{K}\sum_{k=1}^K g \{T_{k, \hat{\theta}}(\theta)\} = 0$, which implies $Kf_{K,\hat\theta}(\theta)w_{K,\hat\theta}(\theta) = 0$. On the other hand, $\sum_{k=1}^K g\{T_{k, \hat{\theta}}(\theta)\} = 0$ implies $g\{T_{k, \hat{\theta}}(\theta)\} = 0$ $\forall$ $k=1, \dots, K$. Among $\{T_{k, \hat{\theta}} \mid k=1,\dots,K\}$, there exists an identity map $T_r$ corresponding to the identity map $S_r$ in $\G_m$ such that $T_r(\theta) = \hat{\theta} + S_r(\theta - \hat{\theta}) = \theta$. Thus $g\{T_r(\theta)\} = g(\theta) = 0$. Thus both sides are zero, and the representation also holds when $\sum_{k=1}^K g\{T_{k, \hat{\theta}}(\theta)\} = 0$. The above argument also shows that $\sum_{k=1}^K g\{T_{k, \hat{\theta}}(\theta)\} > 0$ on the support of $g$.

To show uniqueness on the support of $g$, suppose there exist another $\T_{\G_m, \hat{\theta}}$-invariant pdf $\tilde{f}$ and $K$-component skewing function $\tilde{w}$ such that
\[
g(\theta)=K \tilde{f}(\theta)\tilde{w}(\theta).
\]
Averaging over $T_{k, \hat{\theta}}(\theta)$ for $k=1,\dots,K$, gives
\begin{align*}
f_{K, \hat{\theta}}(\theta) = \frac{1}{K} \sum_{k=1}^K g\{T_{k, \hat{\theta}}(\theta)\}
&= \sum_{k=1}^K \tilde{f}(T_{k, \hat{\theta}}(\theta)) \tilde{w}(T_{k, \hat{\theta}}(\theta)) \\
&= \tilde{f}(\theta) \sum_{k=1}^K \tilde{w}(T_{k, \hat{\theta}}(\theta)) = \tilde{f}(\theta).
\end{align*}
In addition, $\sum_{k=1}^K g\{T_{k, \hat{\theta}}(\theta)\} > 0$ on the support of $g$. Thus, we obtain $\tilde{w}(\theta) = g(\theta)/ \{K\tilde{f}(\theta)\} = w_{K, \hat{\theta}}(\theta)$, and the representation is unique on the support of $g$.
\end{proof}

\subsection{Proposition \ref{prop: K_skew2symm}: Stochastic representations}
\begin{proof}
For (i), let $f(\theta)$ denote the pdf of $\theta=T_{J, \hat{\theta}}(x)$. Note that $T_{k, \hat{\theta}}^{-1}=T_{k, \hat{\theta}}$ and $|dT_{k, \hat{\theta}}(\theta)/d\theta| = 1$. By the law of total probability,
\begin{align*}
f(\theta) &= \sum_{k=1}^K f_{K,\hat\theta}\{T_{k, \hat{\theta}}(\theta)\}
\Pr(J=k\mid x=T_{k, \hat{\theta}}(\theta))
\left|\frac{dT_{k, \hat{\theta}}(\theta)}{d\theta}\right| 
\\
&= \sum_{k=1}^K f_{K,\hat\theta}(\theta)w_{K,\hat\theta}(\theta)
\\
&= Kf_{K,\hat\theta}(\theta)w_{K,\hat\theta}(\theta) 
= g (\theta).
\end{align*}
In the second line, we use the $\T_{\G_m, \hat{\theta}}$-invariant property of $f_{K,\hat\theta}$, where $f_{K,\hat\theta}\{T_{k, \hat{\theta}}(\theta)\}=f_{K,\hat\theta}(\theta)$ for $k=1, \dots, K$. In addition, $\Pr(J=k\mid x=T_{k, \hat{\theta}}(\theta)) = w_{K,\hat\theta}\{T_{k, \hat{\theta}}(T_{k, \hat{\theta}}(\theta))\} = w_{K,\hat\theta}(\theta)$ since $T_{k, \hat{\theta}}^{-1}=T_{k, \hat{\theta}}$. Hence $\theta\sim g$.

For (ii), let $h(x)$ denote the pdf of $x=T_{U, \hat{\theta}} (\theta)$. By the law of total probability,
\begin{align*}
h(x) &= \frac{1}{K}\sum_{k=1}^K g \{T_{k, \hat{\theta}}(x) \} \left|\frac{dT_{k, \hat{\theta}}(x)}{dx}\right|\\
&= \frac{1}{K}\sum_{k=1}^K Kf_{K,\hat\theta} \{T_{k, \hat{\theta}}(x) \} w_{K,\hat\theta} \{T_{k, \hat{\theta}}(x) \} \\
&= f_{K,\hat\theta}(x)\sum_{k=1}^K  w_{K,\hat\theta} \{T_{k, \hat{\theta}}(x) \}
= f_{K,\hat\theta}(x).
\end{align*}
In the first line, $\left|dT_{k, \hat{\theta}}(x)/dx\right| = |S_k| = 1$, while in the third and fourth lines, we used the properties that $f_{K,\hat\theta} \{T_{k, \hat{\theta}}(\theta) \} = f_{K,\hat\theta}(\theta)$ for $k=1,\dots,K$ and $\sum_{k=1}^K w_{K,\hat\theta} \{T_{k, \hat{\theta}}(\theta)\}=1$.    
\end{proof}

\subsection{Theorem \ref{thm:K_optimal_skew_KL}: Optimal skewing function}
\begin{proof}
Let $q^*_{K,\hat\theta} = K\bar q_{K,\hat\theta}(\theta)w^*_{K,\hat\theta}(\theta)$. We have
\begin{align}
\KL(\pi\|q_{K,\hat\theta}) - \KL(\pi\|q^*_{K,\hat\theta})
&= \int \pi(\theta)\log \frac{q^*_{K,\hat\theta}(\theta)}{q_{K,\hat\theta}(\theta)} d\theta 
\nonumber \\
&= \int \pi(\theta)\log \frac{w^*_{K,\hat\theta}(\theta)}{w_{K,\hat\theta}(\theta)} d\theta 
\nonumber \\
&= K\int \bar\pi_{K,\hat\theta}(\theta) w^*_{K,\hat\theta}(\theta)
\log \frac{w^*_{K,\hat\theta}(\theta)}{w_{K,\hat\theta}(\theta)} d\theta 
\label{eq: prop2 pluck-in}\\
&= \sum_{k=1}^K \int \bar\pi_{K,\hat\theta}(\theta)
 w^*_{K,\hat\theta}\{T_{k, \hat{\theta}}(\theta)\}
\log \frac{w^*_{K,\hat\theta}\{T_{k, \hat{\theta}}(\theta)\}}
{w_{K,\hat\theta}\{T_{k, \hat{\theta}}(\theta)\}} d\theta
\label{eq: change var} \\
&=  \int \bar\pi_{K,\hat\theta}(\theta)
\left\{ \sum_{k=1}^K w^*_{K,\hat\theta}\{T_{k, \hat{\theta}}(\theta)\}
\log \frac{w^*_{K,\hat\theta}\{T_{k, \hat{\theta}}(\theta)\}}
{w_{K,\hat\theta}\{T_{k, \hat{\theta}}(\theta)\}} \right\}  d\theta \geq 0.
\label{eq: KLD betw skewing fns}
\end{align}
In \eqref{eq: prop2 pluck-in}, we substitute the unique representation, $\pi(\theta)=K\bar\pi_{K,\hat\theta}(\theta)w^*_{K,\hat\theta}(\theta)$ from Proposition \ref{prop:K_skew_representation}, where $\bar\pi_{K,\hat\theta}(\theta) = \frac{1}{K}\sum_{k=1}^K \pi \{T_{k, \hat{\theta}}(\theta)\}$ is a $\T_{\G_m, \hat{\theta}}$-invariant pdf. In \eqref{eq: change var}, we make the change in variable $\theta \rightarrow T_{k, \hat{\theta}}(\theta)$ for each of the $K$ integrals for $k=1,\dots,K$. Note that $T_{j, \hat{\theta}}^{-1}=T_{j, \hat{\theta}}$ and the Jacobian is one. In addition, as $\bar\pi_{K,\hat\theta}$ is $\T_{\G_m, \hat{\theta}}$-invariant, $\bar\pi_{K,\hat\theta}(T_{k, \hat{\theta}}(\theta)) = \bar\pi_{K,\hat\theta}(\theta)$ for $k=1, \dots, K$. Finally, in \eqref{eq: KLD betw skewing fns}, the term in curly brackets is the KLD between $(w^*_{K,\hat\theta}\{T_{1, \hat{\theta}}(\theta)\},\dots,w^*_{K,\hat\theta}\{T_{k, \hat{\theta}}(\theta)\})^\top$ and 
$(w_{K,\hat\theta}\{T_{1, \hat{\theta}}(\theta)\},\dots,w_{K,\hat\theta}\{T_{k, \hat{\theta}}(\theta)\} )^\top$ for each $\theta$. Each of these vectors can be considered as a probability vector as their entries sum to one due to the property of $K$-component skewing functions. Moreover, $\bar\pi_{K,\hat\theta}(\theta)\ge 0$, and hence 
\[
\KL(\pi\|q^*_{K,\hat\theta}) \le \KL(\pi\|q_{K,\hat\theta}).
\]
Equality holds if and only if the KLD in \eqref{eq: KLD betw skewing fns} is zero almost everywhere, that is, if and only if $w_{K,\hat\theta}\{T_{k, \hat{\theta}}(\theta)\} = w^*_{K,\hat\theta}\{T_{k, \hat{\theta}}(\theta)\}$ for $k=1,\dots,K$ almost everywhere.
\end{proof}

\subsection{Theorem \ref{thm: log_joint_density_skew_K}: Effects of skewness reduction}
\begin{proof} 
For (i), $p(\tilde{b} \mid \theta_g, y) = p_{b \mid \theta_g, y} (\hat{b} + L^{-\top} \tilde{b}) / |L|$ from step 1 of Procedure \ref{Algs:transformation_K}. In addition, $T_{k,c}=T_{k,c}^{-1}$ and $|T_{k,c}| = 1$ for each $k=1,\dots,K$. From step 2 of Procedure \ref{Algs:transformation_K},
\begin{align*}
p(b^* \mid \theta_g, y) 
= \frac{1}{K}\sum_{k=1}^K p_{\tilde{b} \mid \theta_g, y}\{T_{k,c}(b^*)\}= \frac{1}{K|L|}\sum_{k=1}^K p_{b \mid \theta_g, y} [\hat{b} + L^{-\top}\{c+S_k(b^*-c)\} ].
\end{align*}

For (ii), $p(y, \theta^*) = p(b^*, y \mid \theta_g) p(\theta_g)$. Applying the same argument as in (i) gives
\[
p(b^*,y\mid \theta_g)
=\frac{1}{K|L|}\sum_{k=1}^Kp_{b,y\mid \theta_g} [\hat b+L^{-\top}\{c+S_k(b^*-c)\},y ].
\]
Multiplying by $p(\theta_g)$ and taking logarithms gives \eqref{eq:LSC_RVB_density_K}.

For (iii), since $\KL\left\{\phi(\cdot\mid c,I)\| p_{\G}(\cdot\mid \theta_g,y)\right\} = \E_{\phi(z \mid c, I)}\{ \log \phi(z \mid c, I) - \log p_{\G}(z \mid \theta_g,y) \}$, it suffices to show that $\E_{\phi(z \mid c, I)} \log p_{\G}(z \mid \theta_g,y) \geq \E_{\phi(z \mid c, I)}\log p_{\H}(z\mid \theta_g,y)$. Let $u = z - c$ and $f(u) = p_{b\mid \theta_g,y} \{\hat b + L^{-\top}(c + u) \}$ so that $p_{\G}(z\mid \theta_g,y) = \frac{1}{|\G||L|}\sum_{S\in\G} f(Su)$. Since $\H \le \G$, $\G$ can be written as a disjoint union of the right cosets of $\H$ such that $\G=\bigsqcup_{r\in R} \H r$, where $R$ is a set of right coset representatives \citep{Gallian2025}. By Lagrange's theorem,
$|\G|=|\H\backslash \G||\H|=|R||\H|$.
Hence
\begin{align*}
p_{\G}(c+u\mid \theta_g,y)
&=\frac{1}{|L||\G|}\sum_{S\in \G} f(Su) \\
&=\frac{1}{|L||\G|}\sum_{r\in R}\sum_{H\in \H} f(Hru) \\
&=\frac{1}{|R|}\sum_{r\in R} \left[\frac{1}{|L||\H|}\sum_{H\in \H} f\left\{H(ru)\right\} \right] \\
&=\frac{1}{|R|}\sum_{r\in R} p_{\H}(c+ru\mid \theta_g,y).
\end{align*}
As $z\sim \N(c,I)$, $u = z - c \sim N(0, \I)$. Moreover, each $r\in R$ is orthogonal, so $ru \sim \N(0, I)$. By Jensen's inequality,
\begin{align*}
\E_{\phi(z \mid c, I)}\log p_{\G}(z\mid \theta_g,y)
&=\E_{\phi(u)}\log\left\{ \frac{1}{|R|}\sum_{r\in R} p_{\H}(c+ru\mid \theta_g,y) \right\}\\
&\ge\frac{1}{|R|}\sum_{r\in R}\E_{\phi(u)}\log p_{\H}(c+ru\mid \theta_g,y) \\
&=\frac{1}{|R|}\sum_{r\in R}\E_{\phi(u)}\log p_{\H}(c+u\mid \theta_g,y) \\
&=\E_{\phi(z \mid c, I)}\log p_{\H}(z\mid \theta_g,y).
\end{align*}
The inequality is strict unless $p_{\H}(c+ru\mid \theta_g,y)$ is constant $\forall$ $r\in R$. In this case, for any $G \in \G$, 
\begin{align*}
p_{\H}\{T_{G,c}(z) \mid \theta_g, y\} 
&= \frac{1}{|L||\H|}\sum_{H\in \H} f\{ H(c + G(z-c) -c)\} \\
&= \frac{1}{|L||\H|}\sum_{H\in \H} f(HG(z-c)) \\
&= \frac{1}{|L||\H|}\sum_{H\in \H} f(Hr_G(z-c)) \\
&= \frac{1}{|L||\H|}\sum_{H\in \H} f(Hr_0(z-c)) \\
&= \frac{1}{|L||\H|}\sum_{H\in \H} f\{H(z-c)\} = p_{\H}(z \mid \theta_g, y).
\end{align*}
The third line holds because for any $G \in \G$, there exists $r_G \in R$ such that $\H r_G = \H G$. In the special case of the identity element $I \in \G$, there also exist $r_0$ in $R$ such that $\H r_0 = \H$, leading to the 5th equality. Hence equality holds only if $p_{\H}(\cdot\mid \theta_g,y)$ is $\T_{\G, c}$-invariant.

For (iv), taking $\H=\{I\}$ (the trivial group), $p_\H(z \mid \theta_g, y) = \frac{1}{|L|} p_{b \mid \theta_g, y}(\hat{b} + L^{-\top} z)$ which is the pdf of $\tilde{b} \mid \theta_g, y$. Taking $\G=\G_m$ yields the pdf of $b^* \mid \theta_g, y$ in Theorem \ref{thm: log_joint_density_skew_K}(i). Hence, we obtain the KLD inequality in (iv). The equality occurs when $p(\tilde{b} \mid \theta_g, y)$ is $\T_{\G_m, c}$-invariant. 
\end{proof}

\section{Posterior draws of original local variables}
In KNorm-RVB, the original local variables $b$ are mapped deterministically to $\tilde{b}$ via normalization and then stochastically to $b^*$ via skewness reduction. VI is performed on the reparametrized model in terms of $b^*$. For downstream inference, it may be desirable to obtain posterior draws of $b$ from the fitted variational approximation. Procedure \ref{Algs:forward_transformation_K} generates such draws by sampling $\theta^* \sim q(\theta^*)$ and then applying reverse mappings, using Proposition \ref{prop: K_skew2symm}(i) and Theorem \ref{thm:K_optimal_skew_KL}.

Given some blockwise reflection group $\G_m$ and reflection point $c$, Procedure \ref{Algs:transformation_K} treats $\tilde{b}$ as a draw from the $K$-component skew-symmetric density representation of $p_{\tilde b,y\mid \theta_g} (\tilde{b})$ based on Proposition \ref{prop:K_skew_representation}, that is 
\begin{align} \label{eq: kskew approx tilde b}
p_{\tilde b,y\mid \theta_g} (\tilde{b}) = K f_{K,c}(\tilde{b}) w_{K, c}(\tilde{b}),
\end{align}
\begin{align} \label{eq: symm & skew of local cond}
f_{K,c}(\tilde{b}) = \frac{1}{K} \sum_{k=1}^K p_{\tilde b\mid y,  \theta_g} \{T_{k,c}(\tilde{b})\}, \;
w_{K,c}(\tilde{b}) = \frac{p_{\tilde b\mid y, \theta_g} (\tilde{b})}{\sum_{k=1}^K p_{\tilde b\mid y, \theta_g} \{T_{k,c}(\tilde{b})\}} 
= \frac{p_{\tilde b,y \mid \theta_g} (\tilde{b}, y)}{\sum_{k=1}^K p_{\tilde b,y \mid \theta_g} \{T_{k,c}(\tilde{b}), y\}}.
\end{align}
Unlike the $\T_{\G_m, c}$-invariant pdf $f_{K,c}(\tilde{b})$, the $K$-component skewing function $w_{K,c}(\tilde{b})$ is tractable to evaluate because the intractable normalizing constant of $p_{\tilde b\mid y, \theta_g}$ cancels. In VI, we approximate the intractable $f_{K,c}$ by $q(b^*) = \N(\mu_{b^*}, \Sigma_{b^*})$, with $c=\mu_{b^*}$ as reflection point. As $\Sigma_{b^*}$ is a block diagonal matrix with blocks corresponding to each subject, it can be verified that $q(b^*)$ is $\T_{\G_m, c}$-invariant, provided the partitioning of $\G_m$ ensures that all observations for each subject stay within a single block. Theorem \ref{thm:K_optimal_skew_KL} then implies that the optimal $K$-component skewing function is the one given in \eqref{eq: symm & skew of local cond}.

Hence, we obtain a $K$-component skew-symmetric density approximating $p_{\tilde b,y\mid \theta_g} (\tilde{b})$ by replacing the $\T_{\G_m, c}$-invariant pdf $f_{K,c}(\tilde{b})$ in \eqref{eq: kskew approx tilde b} with $q(b^*)$, and a draw from this density then follows from Proposition \ref{prop: K_skew2symm}(i). Finally, we recover $b$ from $\tilde{b}$ by inverting the normalization. This procedure is summarized in Procedure \ref{Algs:forward_transformation_K} for latent Gaussian models and can be extended easily to latent non-Gaussian models by conditioning additionally on the mixing variables and replacing $\hat{b}$, $L$ and $Q$ by $\hat{b}_\NG$, $L_\NG$ and $Q_\NG$.

\begin{algorithm}[htb!]
\captionsetup{name=Procedure}
\caption{Generate posterior draws of original local variables $b$ from $q(\theta^*)$}
\label{Algs:forward_transformation_K}
\begin{algorithmic}[1]
\State Draw $\theta^* = ({b^*}^\top, \theta_g^\top)^\top \sim q(\theta^*)$.
\State For $k=1,\dots,K$, compute 
\[
\Pr(S_*=S_k \mid b^*, \theta_g, y) = w_{K,c}\{T_{k,c}(b^*)\} = \dfrac{p_{\tilde b,y \mid \theta_g} \{T_{k,c}(b^*), y\}}{\sum_{k=1}^K p_{\tilde b,y \mid \theta_g} \{T_{k,c}(b^*), y\}}.
\]
Sample $S_*$ from $\G_m=\{S_1,\dots,S_K\}$ according to above probabilities and set $\tilde b=T_{*,c}(b^*)=c+S_*(b^*-c)$, where $c=\mu_{b^*}$.
\State Invert the normalization: $b=\hat b+L^{-\top}\tilde b$, where $\hat b$ is mode of $p(b\mid\theta_g,y)$ and $L$ is Cholesky factor of $Q+H(\hat b)$.
\end{algorithmic}
\end{algorithm}

\section{Proofs of Section \ref{sec:theory}}
In this section, we present the proofs in Section \ref{sec:theory}.

\subsection{Proof of Theorem \ref{thm:reflection_symmetry}}
\begin{proof}
By Theorem \ref{thm: log_joint_density_skew_K}(i), $p(b^*\mid \theta_g,y)=\frac{1}{K|L|}\sum_{k=1}^Kp_{b \mid \theta_g,y} [\hat b+L^{-\top}\{c+S_k(b^*-c)\}]$.
Hence,
\[
p\{T_{S,c}(b^*)\mid \theta_g,y\} = \frac{1}{K|L|}\sum_{k=1}^Kp_{b \mid \theta_g,y} [\hat b + L^{-\top}\{c+S_kS(b^*-c)\} ]
= p(b^*\mid \theta_g,y), \quad \forall \, S\in\G_m,
\]
because $\{S_kS \mid k=1,\dots,K\} = \G_m$ as $\G_m$ is a group. Hence $p(b^*\mid \theta_g,y)$ is $\T_{\G_m, c}$-invariant about $c$. Since $\G_m$ contains $-I$ for $m \geq 1$, which corresponds to $\varepsilon_1=\cdots=\varepsilon_m=-1$. Taking $S=-I$ gives
\[
p(2c-b^*\mid\theta_g,y) = p(b^*\mid\theta_g,y), 
\]
which implies $p(c - b^* \mid\theta_g,y) = p(c + b^* \mid\theta_g,y)$. Hence $p(b^*\mid\theta_g,y)$ is even-symmetric about $c$. If $\E(\|b^*-c\|\mid\theta_g,y)<\infty$, then using $u=b^*-c$ and a change of variables,
\begin{align} \label{eq: even symm zero mean}
\E(b^*-c \mid \theta_g,y) &= \int u p(c+u\mid\theta_g,y) du \\
&= \int (-u) p(c-u\mid\theta_g,y) du \\
&= -\int u p(c+u\mid\theta_g,y) du = - \E(b^*-c \mid \theta_g,y)
\end{align}
implies $\E(b^*\mid\theta_g,y)=c$. From the Law of total expectation, $\E(b^* \mid y) = \E\{ \E(b^*\mid\theta_g,y) \} = c$ and
\begin{align*}
\Cov(b^*-c,\theta_g \mid y) &= \E((b^*-c) \theta_g^\top \mid y) - 0 \\
&= \E[ \E\{(b^*-c) \theta_g^\top \mid \theta_g, y\} \mid y ] \\
&= \E[ \E\{(b^*-c)  \mid \theta_g, y\}\theta_g^\top \mid y ] 
= 0.
\end{align*}
Thus $\Cov(b^*,\theta_g \mid y) = 0$.
\end{proof}

\subsection{Proof of Theorem \ref{thm:mean_recovery} (Exact recovery of mean)}
\begin{proof}
Since $q_{\nu,\lambda}(\theta^*) = q_{\lambda_{b^*}}(b^* - \nu) q_{\lambda_g}(\theta_g)$, 
\begin{align*}
K_\lambda(\nu) &= \KL\{q_{\nu,\lambda}(\theta^*)\|p(\theta^*\mid y)\} \\
&= \int q_{\nu,\lambda}(\theta^*) \log \frac{q_{\nu,\lambda}(\theta^*)}{p(\theta^*\mid y)} d\theta^* \\
&= \int q_{\lambda_{b^*}}(b^* - \nu) q_{\lambda_g}(\theta_g)\log \frac{q_{\lambda_{b^*}}(b^* - \nu)q_{\lambda_g}(\theta_g)}{p(b^* \mid \theta_g, y) p(\theta_g \mid y)} db^*d\theta_g \\
&= \int q_{\lambda_{b^*}}(\zeta) q_{\lambda_g}(\theta_g)\log \frac{q_{\lambda_{b^*}}(\zeta)q_{\lambda_g}(\theta_g)}{p_{b^*\mid\theta_g,y} (\nu + \zeta) p(\theta_g \mid y)} d\zeta d\theta_g \\
&= - \int q_{\lambda_g}(\theta_g) \left( \int q_{\lambda_{b^*}}(\zeta) \log p_{b^*\mid\theta_g,y} (\nu + \zeta) d\zeta \right) d\theta_g+ C(\lambda), 
\end{align*}
where a change of variable $\zeta = b^* - \nu$ is made in the 4th line and $C(\lambda)$ is independent of the location parameter $\nu$. Differentiating with respect to $\nu$ gives
\begin{align} \label{eq: KL grad wrt nu}
\nabla_\nu K_\lambda(\nu) = - \int q_{\lambda_g}(\theta_g) \left( \int q_{\lambda_{b^*}}(\zeta) \nabla_{b^*} \log p_{b^*\mid\theta_g,y} (\nu + \zeta) d\zeta \right) d\theta_g.
\end{align}
By Theorem \ref{thm:reflection_symmetry}, $p(b^*\mid\theta_g,y)$ is even-symmetric about $c$, that is $p_{b^*\mid\theta_g,y} (c + \zeta) = p_{b^*\mid\theta_g,y} (c - \zeta)$ $\forall \zeta \in \mathbb{R}^N$. It follows that $\nabla_{b^*}\log p(c + \zeta \mid \theta_g, y)= - \nabla_{b^*}\log p(c - \zeta \mid \theta_g, y)$. Therefore $\nabla_{b^*}\log p(c + \zeta \mid \theta_g, y)$ is odd-symmetric about $c$ and it is also an odd function in $\zeta$. At $\nu=c$, the inner integrand in \eqref{eq: KL grad wrt nu} is the product of an odd function of $\zeta$ and the even pdf $q_{\lambda_{b^*}}(\zeta)$. Hence the inner integral over $\zeta \in \mathbb R^N$ vanishes, and $\nabla_\nu K_\lambda(c)=0$. Thus $\nu=c$ is a stationary point for any given $\lambda$.

It remains to show that the stationary point is unique for any given $\lambda$. Define 
\begin{align} \label{eq: S_lambda(nu)}
S_\lambda(\nu) = \int q_{\lambda_g}(\theta_g)  K_{\theta_g, \lambda_{b^*}}(\nu) d\theta_g,
\quad 
K_{\theta_g, \lambda_{b^*}}(\nu) = - \int q_{\lambda_{b^*}}(\zeta)  \log p_{b^*\mid\theta_g,y}(\nu + \zeta) d\zeta, 
\end{align}
and $A = \{\theta_g \mid p(\cdot \mid \theta_g, y) \text{ is strictly concave on a nonempty open set}\}$, where $P(A) > 0$. From Lemma 13 of \cite{Margossian2025a}, $K_{\theta_g, \lambda_{b^*}}(\nu)$ is strictly convex on $\mathbbm{R}^N$ for $\theta_g \in A$. In addition, $K_{\theta_g, \lambda_{b^*}}(\nu)$ is convex on $\mathbbm{R}^N$ for $\theta_g \in A^c$. This is because the concavity of $\log p(b^*\mid\theta_g,y)$ implies that $-\log p_{b^*\mid\theta_g,y}(\nu+\zeta)$ is convex in $\nu$ given $\zeta$. Applying the definition of convexity to  $-\log p_{b^*\mid\theta_g,y}(\nu+\zeta)$  in $\nu$ given $\zeta$, and integrating with respect to $q_{\lambda_{b^*}}(\zeta)$ preserves the convexity. Thus for any $\nu_1 \neq \nu_2$ in $\mathbbm{R^N}$ and $t \in (0,1)$, we have 
\begin{align*}
K_{\theta_g, \lambda_{b^*}}(t \nu_1 + (1-t) \nu_2) &< t K_{\theta_g, \lambda_{b^*}}(\nu_1) + (1-t) K_{\theta_g, \lambda_{b^*}}(\nu_2) \quad \text{for } \theta_g \in A, \\
K_{\theta_g, \lambda_{b^*}}(t \nu_1 + (1-t) \nu_2) &\leq t K_{\theta_g, \lambda_{b^*}}(\nu_1) + (1-t) K_{\theta_g, \lambda_{b^*}}(\nu_2) \quad \text{for } \theta_g \in A^c. 
\end{align*}
Therefore,
\begin{align*}
& S_\lambda(t \nu_1 + (1-t) \nu_2) = \int q_{\lambda_g}(\theta_g) K_{\theta_g, \lambda_{b^*}}(t \nu_1 + (1-t) \nu_2) d\theta_g \\
&= \int_A q_{\lambda_g}(\theta_g) K_{\theta_g, \lambda_{b^*}}(t \nu_1 + (1-t) \nu_2) d\theta_g + \int_{A^c} q_{\lambda_g}(\theta_g) K_{\theta_g, \lambda_{b^*}}(t \nu_1 + (1-t) \nu_2) d\theta_g \\
& < \int_A q_{\lambda_g}(\theta_g) \{t K_{\theta_g, \lambda_{b^*}}(\nu_1) + (1-t) K_{\theta_g, \lambda_{b^*}}(\nu_2)\} d\theta_g \\
& \quad + \int_{A^c} q_{\lambda_g}(\theta_g) \{ t K_{\theta_g, \lambda_{b^*}}(\nu_1) + (1-t) K_{\theta_g, \lambda_{b^*}}(\nu_2) \} d\theta_g \\
&= \int q_{\lambda_g}(\theta_g) \{t K_{\theta_g, \lambda_{b^*}}(\nu_1) + (1-t) K_{\theta_g, \lambda_{b^*}}(\nu_2)\} d\theta_g \\
&= t S_\lambda(\nu_1) + (1-t) S_\lambda(\nu_2).
\end{align*}
Note that the strict inequality holds on the third line because $P(A) > 0$. Hence $S_\lambda(\nu)$ is strictly convex. Since $ C(\lambda)$ is independent of $\nu$, it follows that $K_\lambda(\nu) = S_\lambda(\nu) + C(\lambda)$ is strictly convex in $\nu$ for any given $\lambda$. A strictly convex function has at most one stationary point, so the stationary point at $\nu=c$ is unique.
\end{proof}

\subsection{Proof of Theorem \ref{thm:cor_recovery} (Exact recovery of correlation matrix)}
\begin{proof}
Since $q_{\nu, S, \lambda_g}(\theta^*) = q_{\nu, S}(b^*) q_{\lambda_g}(\theta_g)$, and $\Q = \{q_{\nu, S}\}$ is the location-scale family (Definition \ref{defn: location-scale family}) with a spherically symmetric base pdf $q_0$, we can write $q_{\nu, S}(b^*) = q_0(\zeta) |S|^{-1/2}$, where $\zeta = S^{-1/2} (b^* - \nu)$. Then
\begin{align*}
& \KL\{ q_{\nu, S, \lambda_g}(\theta^*) \| p(\theta^*\mid y)\} 
= \int q_{\nu, S}(b^*) q_{\lambda_g}(\theta_g) \log \frac{q_{\nu, S}(b^*) q_{\lambda_g}(\theta_g)}{p(b^* \mid \theta_g, y) p(\theta_g \mid y)} db^* d\theta_g \\
&= \E_{q_{\lambda_g}(\theta_g)}  \KL\{q_{\nu, S}(b^*) \|  p(b^* \mid \theta_g, y) \}  + \KL\{q_{\lambda_g}(\theta_g) \|  p(\theta_g \mid y)\}.
\end{align*}
Since we are minimizing with respect to $(\nu, S)$ for a fixed $\lambda_g$, the second term is a constant and we can focus only on the first term. Making a change of variable, $\zeta = S^{-1/2} (b^* - \nu)$,
\begin{align*}
\KL\{ q_{\nu, S}(b^*) \|  p(b^* \mid \theta_g, y) \} &= \int q_0(\zeta) \log \frac{q_0(\zeta) |S|^{-1/2}}{ p(S^{1/2} \zeta + \nu \mid \theta_g, y) } d\zeta \\
&= -\H(q_0)  -\frac{1}{2} \log|S| - \int q_0(\zeta) \log p_{b^* \mid \theta_g, y}(S^{1/2} \zeta + \nu) d\zeta,
\end{align*}
where $\H(q_0) = - \int  q_0(\zeta) \log q_0(\zeta) d\zeta$ denotes the entropy. Hence it suffices to find  $(\nu, S)$ that minimizes
\begin{align*}
f_{\lambda_g} (\nu, S) = -\frac{1}{2} \log|S| - \E_{q_{\lambda_g} (\theta_g) q_0(\zeta)} \{ \log p_{b^* \mid \theta_g, y}(S^{1/2} \zeta + \nu) \}.
\end{align*}
Observe that the second term is similar to $S_\lambda(\nu)$ in \eqref{eq: S_lambda(nu)}, where the only differences are that $q_{\lambda_{b^*}}(\zeta)$ is replaced by $q_0(\zeta)$ and $\zeta$ is replaced by $S^{1/2} \zeta$ in the argument of $\log p_{b^* \mid \theta_g, y}(\cdot)$. As in Theorem \ref{thm:mean_recovery}, $q_0$ is even-symmetric about 0 as it is spherically symmetric, and $ \log p(b^*\mid\theta_g,y)$ satisfy the same conditions of being concave on $\mathbb R^N$ almost everywhere (a.e.) with respect to $q_{\lambda_g}$ and strictly concave on a nonempty open set with positive probability under $q_{\lambda_g}$. Therefore we can conclude that the second term is strictly convex with respect to $\nu$ and it has a unique stationary point at $\nu = c$ for any given $\lambda_g$ as in Theorem \ref{thm:mean_recovery}.

Setting $\nu=c$, our objective function simplifies to 
\begin{align*}
f_{\lambda_g} (c, S) = -\frac{1}{2} \log|S| - \int q_{\lambda_g} (\theta_g) \left( \int q_0(\zeta) \log p_{b^* \mid \theta_g, y}(S^{1/2} \zeta + c) d\zeta \right) d\theta_g.
\end{align*}
Let $A$ denote the set of $\theta_g$ where $p(b^*\mid\theta_g,y)$ is log-concave and elliptically symmetric about $c$ with scale matrix $M$. Then there exists a spherically symmetric pdf $p_{0, \theta_g}$ such that $p_{b^* \mid \theta_g, y}(b^*) = p_{0, \theta_g}\{M^{-1/2}(b^*-c)\}|M|^{-1/2}$ for each $\theta_g \in A$. Substituting this in $f_{\lambda_g}(c, S)$, we have 
\begin{align*}
f_{\lambda_g} (c, S) = f_{\lambda_g} (J) = -\log|J|  - \int_A q_{\lambda_g} (\theta_g) \int q_0(\zeta) \log p_{0, \theta_g} (J\zeta) d\zeta d\theta_g,
\end{align*}
where $J = M^{-1/2} S^{1/2}$ since the integral over $A^c$ is equal to zero. It suffices to minimize $f_{\lambda_g} (J)$ with respect to $J$. Note that $f_{\lambda_g} (J)$ is strictly convex in $J$ and hence any stationary point will correspond to a unique global minimizer. Following \cite{Margossian2025a}, we will show that such a stationary point occurs at $J = \gamma I$ for some $\gamma > 0$. The strict concavity property of $f_{\lambda_g} (J)$ follows from the strict concavity of $\log|J|$ and the assumptions that $p(b^* \mid \theta_g, y)$ and hence $p_{0, \theta_g}$ are log-concave for $\theta_g \in A$.

Since $p_{0, \theta_g}$ and $q_0$ are spherically symmetric, we can define the functions $g_{1, \theta_g}$ and $g_2$ such that $\log p_{0, \theta_g}(J\zeta) = g_{1, \theta_g}(\|J\zeta\|)$ and $q_0(\zeta) = g_2(\|\zeta\|)$. Then 
\begin{align*}
 f_{\lambda_g} (J) &= -\log|J| - \int_A q_{\lambda_g} (\theta_g) \int g_2(\|\zeta\|) g_{1, \theta_g}(\|J\zeta\|) d\zeta d\theta_g, \nonumber \\
 \nabla_{\vec(J)}  f_{\lambda_g} (J) &= -\vec\left\{ J^{-\top} + \int_A q_{\lambda_g} (\theta_g) \left( \int g_2(\|\zeta\|) g_{1, \theta_g}'(\|J\zeta\|) \frac{J\zeta \zeta^\top}{\|J\zeta\|} d\zeta \right) d\theta_g \right\}.
 \end{align*}
Suppose a minimizer exists at $J = \gamma I$ for some $\gamma>0$, then the gradient is zero at this point, which means
\begin{align*}
\gamma^{-1} I &= - \int_A q_{\lambda_g} (\theta_g) \left( \int g_2(\|\zeta\|) g_{1, \theta_g}'(\gamma\| \zeta\|) \frac{\zeta \zeta^\top}{\|\zeta\|} d\zeta \right) d\theta_g 
\end{align*}
\cite{Margossian2025a} show that the expression in () is a scalar multiple of the identity matrix in their Theorem 10's proof. As both sides of the equations are scalar multiples of the identity matrix, we can solve for $\gamma$ by equating their traces. As $\sum_i \zeta_i^2 = \|\zeta\|^2$, We have
\begin{align*}
N / \gamma &= - \int_A q_{\lambda_g} (\theta_g) \left( \int g_2(\|\zeta\|) g_{1, \theta_g}'(\gamma \| \zeta\|) \|\zeta\| d\zeta \right) d\theta_g \\
&= \int_A q_{\lambda_g} (\theta_g) \left(  - \frac{2\pi^{N/2}}{\Gamma(N/2)} \int g_2(r) g_{1, \theta_g}'( \gamma r)  r^N dr \right) d\theta_g, 
\end{align*}
where the last line is evaluated using spherical coordinates as in \cite{Margossian2025a}. In addition, they showed that the term in the last line in parentheses is a positive increasing function of $\gamma$. This holds for each $\theta_g \in A$, and integration with respect to $q_{\lambda_g}$ preserves this property. Hence we conclude that a unique solution $J = \gamma I$ exists, which implies $S = \gamma^2 M$. Note that $\gamma$ depends on $\lambda_g$. Since $q_0$ is spherically symmetric, $\Cov_{q_0}(\zeta)=\kappa I$ for some $\kappa>0$, so at the minimizer $\Cov_{q_{\nu, S}}(b^*)=\kappa S=\kappa\gamma^2M$ and the correlation matrix implied by $M$ is recovered exactly.
\end{proof}

\end{document}